\PassOptionsToPackage{colorlinks=true, linkcolor=blue, citecolor=blue, urlcolor=blue, anchorcolor=blue}{hyperref}
\documentclass[a4paper]{iopjournal}

\usepackage[utf8]{inputenc}
\usepackage[english]{babel}
\usepackage[T1]{fontenc}
\usepackage{amsmath,amssymb,bm,amsfonts,mathrsfs,bbm}
\usepackage{tikz}
\usetikzlibrary{patterns, calc, decorations.pathreplacing, shapes, intersections, backgrounds}
\usepackage{tikz-3dplot}
\usepackage[colorlinks=true, linkcolor=blue, citecolor=blue, urlcolor=blue, anchorcolor=blue]{hyperref}

\usepackage{orcidlink}
\usepackage{lipsum}
\usepackage{amsthm}
\usepackage{subcaption}
\usepackage{comment}
\usepackage{diagbox}
\usepackage{xcolor}
\usepackage{soul}
\usepackage{dsfont}
\usepackage{graphicx}
\usepackage{physics}
\usepackage[numbers]{natbib}
\usepackage{wrapfig}
\usepackage{float} 
\usepackage{ragged2e}
\newtheorem{theorem}{Theorem}

\newtheorem{corollary}{Corollary}[theorem]
\newtheorem{definition}{Definition}

\newtheorem{example}{Example}
\newtheorem{proposition}[theorem]{Proposition}


\newcommand{\R}{\mathbb{R}} 
\newcommand{\mc}[1]{\mathcal{#1}} 
\newcommand{\mk}[1]{\mathfrak{#1}} 
\newcommand{\ms}[1]{\mathsf{#1}} 
\newcommand{\tx}[1]{\mathrm{ \ {#1} \ }} 
\newcommand{\then}{\Rightarrow} 
\newcommand{\ra}{\rightarrow} 

\begin{document}

\articletype{Paper} %

\title{Reachable sets under Kolmogorov dynamics in the  probability simplex}

\author{Alfonso Fernández de Bobadilla$^1$\orcid{0009-0003-7361-1849}, Mykhailo Hontarenko$^{1,2}$\orcid{0009-0005-1354-7149},  Guillem Müller-Rigat$^1$\orcid{0000-0003-0589-7956
}  and \\  Karol {\.Z}yczkowski$^{1,3}$\orcid{0000-0002-0653-3639}}

\affil{$^1$Institute of Theoretical Physics, Jagiellonian University, ul. {\L}ojasiewicza 11, Krak{\'o}w, 30-348, Poland}

\affil{$^2$Doctoral School of Exact and Natural Sciences, Jagiellonian University, ul. Łojasiewicza 11, 30-348 Kraków, Poland}

\affil{$^3$Center for Theoretical Physics, Polish Academy of Sciences, Al. Lotników 32/46, 02-668 Warszawa, Poland}





\begin{abstract}
\justifying
The finite-time reachability problem -- whether two probability vectors can be transformed in a time $t$ through a generator from a given set -- is a key question in Markovian dynamics. We address this problem for normalized Kolmogorov generators with tools from differential geometry. The shortest time $T$ connecting an ordered pair of states defines an (asymmetric) quasi-distance in the probability simplex, which distinguishes between outgoing and incoming reachable sets, and captures the inherent irreversibility of Markov dynamics. The Finsler metric corresponding to normalized generators with bounded diagonal entries is characterized, and time-independent evolutions that saturate speed limits are presented. Analytical boundaries of both reachable sets are derived and the ratio of their corresponding volumes 
assess 
the bounds obtained. The asymmetry with respect to the uniform probability vector is shown to be related to the Kullback--Leibler relative entropy.
\end{abstract}

\section{Introduction}
\justifying
The evolution of classical states -- $N$-point probability vectors -- under Markov dynamics is governed by the Kolmogorov equation~\cite{Kolmogoroff1931}. It has broad applications across many disciplines~\cite{Anderson1991}, including models of neuronal activity~\cite{Buice2010}, particle diffusion processes~\cite{Einstein1905,Smoluchowski1906}, and in the early quantum theory, where it was used to describe driven dissipative harmonic oscillators in radiation fields~\cite{Planck1917}. These ideas later became foundational in quantum optics, where they play a central role in the description of open quantum systems~\cite{Carmichael1999,Gardiner2004-va}.

A particularly important question in Markovian dynamics is the reachability problem, i.e. to characterize the set of states that can be reached from a given initial one. To date, this problem has been mostly considered in the context of controllability of dynamical systems~\cite{Sontag1998, khaneja2000}, which ask whether a state can be produced or not given driving inputs under physically-inspired constraints, e.g. forbidding certain state transitions. Reachable subsets of the probability simplex under Markovian dynamics were recently discussed in the context of thermal resources \cite{lostaglio2018elementary,lostaglio2022continuous}. Numerous techniques from convex analysis~\cite{Borkar1990}, Lie wedges and semigroup orbits ~\cite{wedges1,wedges2,wedges3}, graph theory ~\cite{Elamvazhuthi2021}, and dynamical programming~\cite{Avila2022} have been exploited to address the question. However, these approaches either consider discrete time-dynamics in a Markov chain or focus on the absolute reachability of the target (in an arbitrarily large time). Nonetheless, time-constrained reachable sets have recently attracted attention, for example in the context of covariance control~\cite{Liu2024}, or in the classical counterparts of quantum speed based on entropy production~\cite{Gu2023SpeedLimit, Shanahan2018}.

In this work, we consider the finite-time reachability problem in the probability simplex which, to be well posed, requires bounding the generators (as if not, arbitrary pairs of states may be reached arbitrarily fast). To this end, we inspect dynamics whose Kolmogorov generators satisfy $\|Q\|\leq 1$ for some chosen norm, and ask for the evolution transforming an initial state $\bf p$ into a final state ${\bf q}$ in the shortest possible time. Note that the norm of $Q$ can be interpreted as a restricted resource (e.g.  mean energy or average coupling strength) one wishes to optimize. In this resource-constrained setting, the minimum time to join pairs of states becomes a quasi-distance, which complements other divergences on the probability simplex ~\cite{stojmirovic2008quasi,eckstein2025maxtypequasimetricsprobabilitysimplices}. Unlike a standard symmetric distance, this quasi-distance distinguishes between outgoing and incoming sets, i.e. the sets of states that are reachable from or can reach a given initial probability vector. These sets are different due to the inherent irreversibility of Kolmogorov dynamics, and moreover, both sets can be nonconvex. 

Furthermore, the resulting time quasi-distance is shown to be compatible with an underlying Finsler metric, whose form is determined by the particular norm $\|\cdot\|$ constraining the dynamics. This observation allows a geometric approach to the reachability problem, with reachable sets and optimal trajectories becoming, respectively, metric balls and geodesic curves of the time quasi-distance. We explicitly provide the time quasi-distances that solve the reachability problem when the $\ell_p$ norm on the diagonal rates of the admissible Kolmogorov generators is bounded, and show that the reachable sets corresponding to the cases of bounded average and maximum diagonal rates constitute, respectively, inner and outer bounds to the reachable sets due to dynamics bounded in any other permutation-invariant norm on the diagonal of the generator, not necessarily of the diagonal-$\ell_p$ type. Importantly, the bounds are formulated in the general case of time-inhomogeneous Markov dynamics, i.e. when the generator may depend explicitly on time. Finally, one may translate the reachability problem with bounded generators into a task of redistributing water between $N$ tanks in an optimal way by using pumps with bounded draining power. This model will be discussed further in Section~\ref{sec:preliminaries} and Appendix \ref{app: toy}.

As an appetizer, we compare 
 in Fig.~\ref{fig:riemann_finsler} the time distance $(a)$ -- 
approximately  symmetric car driving time,
with  the quasi-distance $(b)$ --
asymmetric hiking time in the mountains,  and its corresponding reachable sets.

\begin{figure}[h]
    \centering
      \begin{subfigure}{0.45\textwidth}
        \centering
        \includegraphics[width=\linewidth]{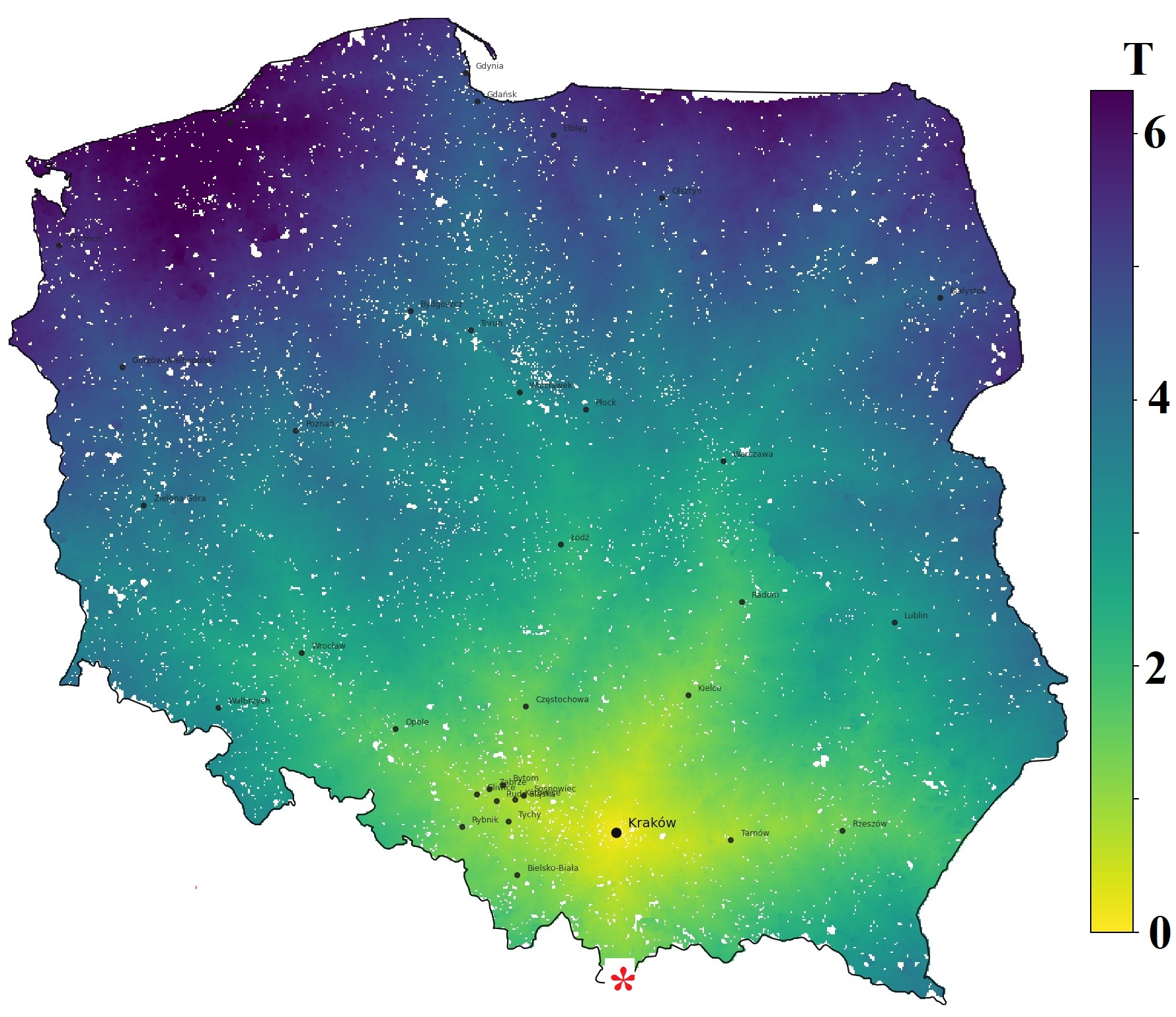}
        \caption{}
    \end{subfigure}
    \hfill
    \begin{subfigure}{0.48\textwidth}
        \centering
        \includegraphics[width=0.80\linewidth]{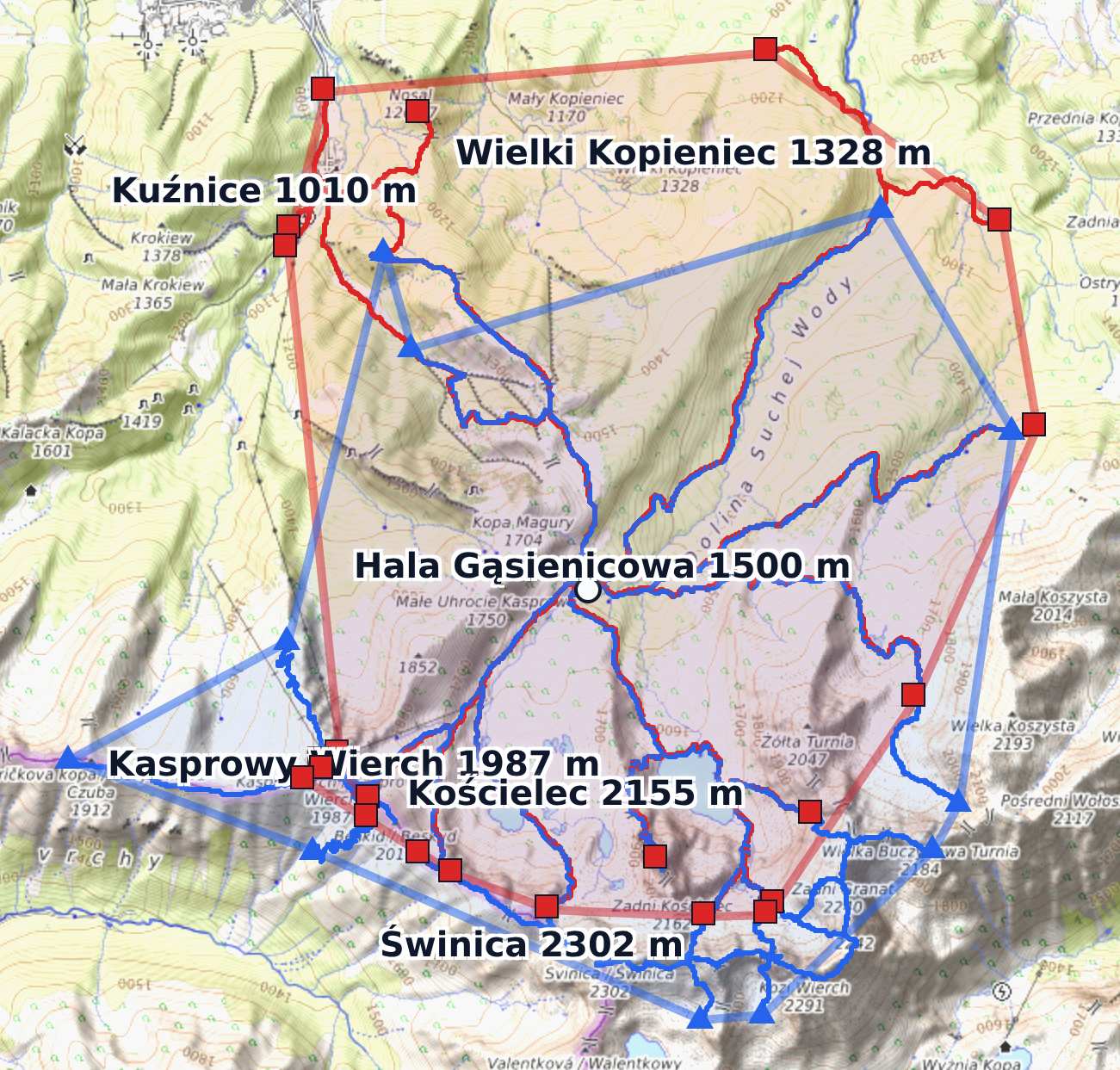}
        \caption{}
    \end{subfigure}
   \caption{$(a)$ Exemplary {\sl time distance}: driving time $T$ in Poland
   from Kraków (in hours) according to Open Street Map -- the balls corresponding to a fixed $T$, represented by different colors, are not convex;
    $(b)$ 
   {\sl time quasi-distance}: the 
   outgoing set (plotted in red)  -- points in Tatra mountains (red star in panel a) 
   reachable from G{\k a}sienicowa hut (white point)
   in 90 minutes
    according to the data from portal Sharpmap.com 
    differs from  the incoming set 
   (points in blue), from which one can reach the hut in 90 minutes.
   Note the differences in altitude which explain this asymmetry.
   } 
    \label{fig:riemann_finsler}
\end{figure}

The work is structured as follows. In Section~\ref{sec:preliminaries}, the time-metric and time-quasi distance induced by the generator norm are introduced, along with the necessary mathematical definitions. In Section~\ref{sec:results}, we present the main results. In particular, an explicit characterization of the time quasi-distances induced by dynamics whose diagonal entries are $\ell_p$-bounded, as well as inner and outer bounds applying in the more general setting of permutation-invariant norms on the diagonal of the generator, holding in arbitrary dimension and for time inhomogeneous dynamics. We also discuss the tightness of the obtained bounds and relate the corresponding quasi-distances to entropic quantities. Finally, in Section~\ref{sec:conclusions} we present our conclusions and directions for future work. An analogous water-distribution model is described in Appendix \ref{app: toy} and proofs are presented in Appendices \ref{app:prelim} and \ref{app:results}.

\section{Preliminaries}
\label{sec:preliminaries}

Some introductory definitions and results regarding the geometries of the spaces of classical states and classical Markov evolutions are reviewed, and the Finsler time-metric and time quasi-distance on the simplex are introduced.

\subsection{Markov dynamics on the simplex}

We start setting the scene by defining state space.

\begin{definition}[States, state space]
    A state of an $N$-level system is given by an $N$-point probability vector,
    \begin{equation}
        {\bf{p}} = [ \, p_1,...,p_{N} \, ]^T \tx{ s.t. } p_i\geq 0 \tx{ and }\sum_{i=1}^N p_i = 1,
    \end{equation}
    A state ${\bf{p}}$ is said to be interior iff $p_i>0 \ \forall i$. We call the simplex of all probability vectors,
    \begin{equation}
        \Delta_{N-1} = \left\{ {\bf{p}}\mathrm{\ s.t.\ } p_i\geq 0,\sum_{i=1}^Np_i=1\right\},  
    \end{equation}
    the state space of the system. Denote as $\omega=[1/N,...,1/N]^T$ the maximally-mixed state.
\end{definition}

Geometrically, the simplex $\Delta_{N-1}$ is a compact smooth manifold with corners \cite{lee2013smooth, Cai2025}. From an interior state one may travel in any direction, while at the boundary only a certain cone of directions is available. Besides this, one can define the tangent space in the usual manner.

\begin{definition}[Tangent vectors, tangent space]
    Given a state ${\bf{p}}$, its tangent space $T_{\bf{p}}\Delta_{N-1}$ is defined as the collection of all tangent vectors ${\bf{v}}=d\mathbf{p}(t)/dt|_{t=0}$ at ${\bf{p}}={\bf{p}}(0)$ of directed differentiable curves $\mathbf{p}(t)$ on the simplex, 
    \begin{equation}
        T_{\bf{p}}\Delta_{N-1} = \left\{{\bf{v}} \tx{ s.t. } \sum_{i=1}^N v_i=0 \tx{ and } v_j\geq 0 \tx{ when } p_j=0\right\}.
    \end{equation}
\end{definition}

\

As a final remark on the geometry of $\Delta_{N-1}$, let us introduce the following partition of the simplex into disjoint chambers, which will appear often throughout the work. Given state ${\bf{p}}$, one may classify other points ${\bf{q}}$ into chambers according to the signs of $q_i-p_i$. Each such region is uniquely characterized by the sets of indices $I_-,I_+\subset\{1,...,N\}$ for which $q_i-p_i< 0$ and $q_i-p_i \geq 0$, respectively. Similarly, one may partition the tangent space at some state into disjoint regions, according to the signs of entries of tangent vectors, see Fig. \ref{fig:chambers}. 
Several properties derived in this work strongly depend on the choice of one of these chambers.

\begin{figure}[h!]
    \centering
    \includegraphics[width=0.75\linewidth]{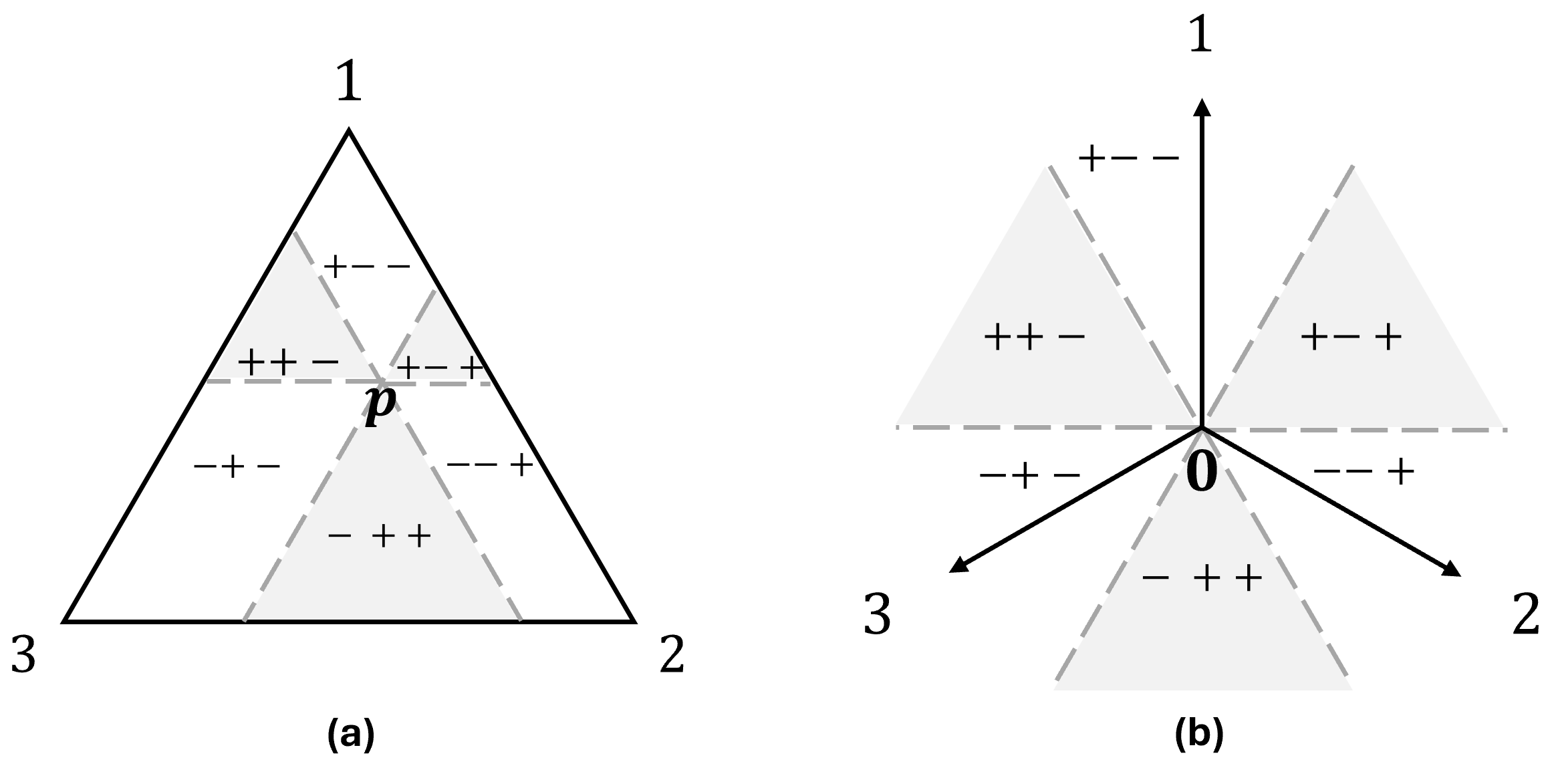}
    \caption{Exemplary partition of the state space (a) and tangent space (b) for $N=3$. For a given ${\bf{p}}$, other states ${\bf{q}}$ on the simplex can be grouped into disjoint chambers according to the signs of the entries of ${\bf{q}}-{\bf{p}}$ (a). Similarly, the tangent space at ${\bf{p}}$ can be partitioned according to the signs of the entries of ${\bf{v}}$ (b). Regions where only a single entry is negative are colored in light gray. Here, $1,2,3$ represent the deterministic states $[1,0,0], [0,1,0], [0,0,1]$ respectively.}
    \label{fig:chambers}
\end{figure}

Time evolution of the system is described by stochastic maps $\mc{E}_{t_2t_1}$, which satisfy $\sum_{i}[\mc{E}_{t_2,t_1}]_{ij}=1$ and $[\mc{E}_{t_2t_1}]_ {ij}\geq 0$, sending initial states ${\bf{p}}$ at time $t_1$ to later states $\mc{E}_{t_2t_1}{\bf{p}}$ at time $t_2$ \cite{Rivas2012}. In this work, we will be concerned with time-continuous Markovian evolutions, both time-homogeneous (when $\mc{E}_{t_2,t_1} = \mc{E}_{t_2-t_1,0} \equiv \mc{E}_{t_2-t_1}$) and time in-homogeneous (if not).

\begin{definition}[Time-homogeneous Markov dynamics]
    A family of time-homogeneous stochastic maps on the simplex $\{\mathcal{E}_{t} \}_{t\geq 0}$ fulfilling $\mc{E}_{t+s} = \mc{E}_t \circ \mc{E}_s$ (Markov condition) and $\mc{E}_0 = \tx{id}$ (with $\tx{id}$ the identity map) forms a Markov semigroup. Such semigroups have an associated Kolmogorov generator $Q$ obeying
    \begin{equation}
    \label{eq:Kolconditions}
        \sum_{i=1}^N Q_{ij}=0 \tx{ and } Q_{ij}\geq 0 \tx{ for } i \neq j,
    \end{equation}
    in terms of which $\mc{E}_t=e^{tQ}$ and the system obeys the master equation
    \begin{equation}\label{eq:classmaster}
        \frac{d{\bf p}}{dt} = Q {\bf p}.
    \end{equation}
\end{definition}
We may also use overdot, 
$\dot{\mathbf{p}}$, to indicate the time derivative. Eq. ~\eqref{eq:classmaster} implies that the off-diagonal entries $Q_{ij}$ represent individual decay rates of probability inflow from configuration $j$ into configuration $i$.  The column sum-zero condition,  
$Q_{ii}=-\sum_{j\neq i}Q_{ij}$,
gives the total outflow from configuration $i$ into the rest, a sort of \textit{collective decay rate} describing the net rate of population loss. Hence, the Kolmogorov matrix $Q$ is often known as the \textit{transition-rate matrix}, or simply \textit{rate matrix}. In the time-homogeneous scenario these transition rates are fixed, but one may also allow the Kolmogorov matrix to depend on time.

\begin{definition}[Time-inhomogeneous Markov dynamics]
    A system evolving under~\eqref{eq:classmaster} with time-dependent Kolmogorov generator $Q_t$
    which satisfies $\sum_{i} Q_{t,ij}=0$ and $Q_{t,ij}\geq 0$ for $i\neq j$ at all times and is continuous in $t$, is said to be undergoing time-inhomogeneous Markov dynamics. The formal solution to this time-inhomogeneous master equation is no longer of semigroup type, but
    \begin{equation}
        {\bf{p}}(t) = \mc{T}e^{\int_0^t Q_s ds}{\bf{p}}.
    \end{equation}
    where $\mc{T}$ is the time-ordering symbol \cite{Anderson1991}.
\end{definition}

The particular form of the Kolmogorov matrix will determine -- through the master equation -- the future behavior of the system. Hence, it is vital for our purposes to have a good understanding of the geometry and structure of the space of such matrices.

\begin{definition}
    Denote by $\mk{gen}$ the set of all possible Kolmogorov rate matrices for an $N$-level system.
    \begin{equation}
        \mk{gen} := \{Q \tx{ s.t. } Q_{ij}\geq 0 \tx{ for } i\neq j \tx{ and } \sum_{i=1}^{N}Q_{ij}=0\}
    \end{equation}
    We introduce the following two physically motivated relations on $\mk{gen}$, allowing to compare pairs of generators
    \begin{itemize}
        \item $Q\leq_I Q' \! \iff \! Q_{ij} \leq Q_{ij}' \ \forall i\neq j$~, defined by comparison of individual decay rates.

        \item $Q\leq_C Q' \! \iff \! |Q_{ii}|\leq |Q_{ii}'| \ \forall i$, defined by comparison of collective decay rates.
    \end{itemize}
\end{definition}
It is not difficult to establish the following observation, proven in Appendix~\ref{app:prelim}.
\begin{proposition}\label{th: gen}
    The set $\mk{gen}$ forms a simplicial cone in $\mathbb{R}^{N(N-1)}$ with extremal rays spanned by $\{E_{(ij)}:=E_{ij} - E_{jj}\}_{i\neq j}$, where $E_{ij}$ is the matrix whose only nonzero entry is $1$ in position $(i,j)$. The relation $\leq_I$ is the natural partial order of $\mk{gen}$ as a cone, while $\leq_C$ is not a partial order but just a pre-order, since $Q\leq_C Q'$ and $Q'\leq_C Q \nRightarrow Q = Q'$. In any case, they satisfy $Q \leq_I Q' \then Q\leq_C Q'$.
\end{proposition}


Note that, in general, many matrices in $\mk{gen}$, when acting on a given state $\textbf{p}$, will push it with the same given velocity $\textbf{v}\in T_\textbf{p}\Delta_{N-1}$.


\begin{definition}
    The set of Kolmogorov matrices that act on the state ${\bf{p}}\in \Delta_{N-1}$ to give the tangent vector ${\bf{v}}\in T_{\bf p}\Delta_{N-1}$ is $\mk{gen}({\bf{p}},{\bf{v}}):=\{Q\in\mk{gen}\mathrm{\  s.t.\ } Q\bf{p} = v\}$.
\end{definition}

We are particularly interested characterizing these sets, as they collect all Kolmogorov generators driving a given state with same initial velocity. It is easy to see that $\mk{gen}({\bf{p}},{\bf{v}})$ is a convex and unbounded affine subset of $\mk{gen}$, of dimension $(N-1)^2$. The following two Propositions shed more light into their structure.


            

\begin{proposition}\label{th: gen(p,v)1}
    For given interior ${\bf{p}},{\bf{q}}\in \Delta_{N-1}$ and fixed ${\bf{v}}$, the sets $\mk{gen}({\bf{p}},{\bf{v}})$ and $\mk{gen}({\bf{q}},{\bf{v}})$ are linearly isomorphic.
\end{proposition}

\begin{proposition}\label{th: gen(p,v)2}
    Let ${\bf{v}}\in T_{\bf{p}}\Delta_{N-1}$ for an interior state $\mathbf{p}\in \Delta_{N-1}$. For $\mathbf{v}\neq \mathbf{0}$, minimal elements $Q$ of $\mk{gen}(\mathbf{p},\mathbf{v})$ with respect to $\leq_C$ form an $(n^+-1)(n^--1)$-dimensional polytope, and have the form:
    \begin{equation}
        Q_{ij} =
        \begin{cases}
            -\delta_{ij}|v_i|/p_i & \tx{ if } v_i,v_j\leq 0\\
            W_{ij}|v_j|/p_j & \tx{ if } v_i> 0, v_j\leq 0\\
            0 & \tx{ otherwise}
        \end{cases}
    \end{equation}
    where $W_{ij}$ is entrywise non-negative and satisfies $\sum_{i\in I_+}W_{ij} = 1$ and $\sum_{j\in I_-}W_{ij}v_j=-v_i$. For $\mathbf{v}=0$, its minimal element is the zero matrix.
\end{proposition}

\noindent For the proof of both Propositions see Appendix~\ref{pr: gen(p,v)}.

\

While the sets $\mk{gen}({\bf{p}},{\bf{v}})$ may not have a minimal element in the order $\leq_I$, according to Proposition \ref{th: gen(p,v)2}, if we order matrices by comparison of collective decay rates there is now a distinguished set of \textit{smallest generators}, characterized by the above conditions.
As closing remark on evolutions, let us highlight a particular example of Markovian evolution that will play a relevant role in this work.
\begin{example}
    We refer as replacer evolution of strength $\kappa>0$ and target ${\bf{q}}\in\Delta_{N-1}$~\cite{Cai2025} to the map
    \begin{equation}
    \label{eq:dynclassrepl}
        {\bf{p}} \mapsto {\bf{p}}(t) = \mc{E}_t({\bf{p}}) := e^{-\kappa t}\, {\bf{p}} + (1-e^{-\kappa t}){\bf{q}},
    \end{equation}
    which can be seen to form a Markov semigroup, since $\mc{E}_t\circ\mc{E}_s = \mc{E}_{t+s}$ and $\mc{E}_0=\mathrm{id}$. Its associated generator is $Q_{ij} = \kappa (q_i - \delta_{ij})$, inducing the dynamics
    \begin{equation}
    \label{eq:classrepl}
        \frac{d{\bf p}}{dt} = \kappa({\bf q}-{\bf p}).
    \end{equation}
\end{example}

\subsection{The finite-time reachability problem}

We move now to the main question of the work. The reachability problem is central to Markov dynamics and control theory \cite{Sontag1998,Elamvazhuthi2021}. In this paper, we are interested in the set of states that can be reached from a given initial one via (possibly time-inhomogeneous) Markov dynamics \textit{during some finite time} $t$. Note that this question is meaningless if arbitrary evolutions are allowed (as any state would be reachable arbitrarily fast) -- we must impose constraints on how large the generators can be.

One can do this bounding the maximum ``size'' of the generator which, from a physical perspective, means that the resources allotted for the dynamics are bounded and we wish to make an optimal use of them. A sensible way to implement this size constraint is to bound some (physically inspired) functional $\| \cdot \| : \mk{gen} \rightarrow \R^+$, of which it is reasonable to demand
\begin{itemize}
    \item Positivity: $\forall Q \in \mk{gen}, \ \|Q\|=0 \then Q=0$. 

    \item Homogeneity: $\forall Q \in \mk{gen}, \ \|\alpha Q\|=\alpha\|Q\| \ \forall \alpha > 0$ 

    \item Subadditivity: $\forall Q, Q' \in \mk{gen}, \ \|Q+Q'\| \leq \|Q\| + \|Q'\|$
\end{itemize}
That is, we demand that, on the cone $\mk{gen}$, $\|\cdot\|$ obeys the axioms of a norm\footnote{Note that on the full vector space $\mathrm{span}(\mk{gen})\simeq \R^{N(N-1)}$, $\|\cdot\|$ is not required to be a norm, just on the cone $\mk{gen}$. In fact, the family of norms considered in this work are \textit{seminorms} on this larger space, as there exist non-zero matrices $Q\notin \mk{gen}$ such that $\|Q\| = 0$. For details on norms on cones, see \cite{WardvanGaans1999}.}. For simplicity and physical insight, in this work we choose to work with a particular family of norms measuring generators according to the sizes of their collective -- i.e. diagonal -- decay rates.

\begin{definition}[Isotropic norm]
\label{def:iso}
    An isotropic norm is a norm $\|\cdot \|$ for the cone of Kolmogorov matrices that just depends on the collective decay rates $\{Q_{ii}\}_{i=1}^N$, is monotone under $\leq_C$ and invariant under their permutation. All such norms are expressible as $\|Q\| = f(Q_{11},...,Q_{NN})$ for some monotonic, permutation-invariant, convex and positive homogeneous function $f$.
\end{definition}
To any isotropic norm $\|\cdot\|$, one can assign a corresponding ``time-scale'' -- giving a characteristic time of the dissipation processes allowed by the constraint $\|Q\|\leq 1$ -- defined as the number $\tau>0$ such that its corresponding function $f$ (see Definition~\ref{def:iso}) satisfies
\begin{equation}\label{eq: time-scale}
    f(x,0,...,0) = \tau |x|.
\end{equation}
The existence of such positive $\tau$ follows from positive-homogeneity of $f$, and by permutation-invariance it is unimportant which entry is left non-zero in \eqref{eq: time-scale}. An illustrative family of isotropic norms are $\ell_\alpha$-norms on the diagonal of the generators, $\|Q\|^\tx{diag}_\alpha:=\tau\left(\sum_i|Q_{ii}|^\alpha\right)^{1/\alpha}$. Note that, for any $\alpha\leq \alpha'$, the following inequality holds $\|Q\|^\tx{diag}_{\alpha}\geq\|Q\|_{\alpha'}^\tx{diag}$. The extremal cases,
\begin{align}
    \|Q\|^\tx{diag}_{1} := \tau\sum_i |Q_{ii}| \ \tx{ and } \ \| Q\|^\tx{diag}_\infty := \tau\max_{i}|Q_{ii}|.
\label{eq:norms}
\end{align}
are of special relevance for the main results of the paper. \\

Inspired by Ref. \cite{COK23}
we illustrate an analogy between normalized Kolmogorov dynamics and the redistribution of water among different water tanks. The state of the system, $\mathbf{p}\in\Delta_{N-1}$, can be interpreted as representing the water levels in $N$ water tanks. Under this analogy, the Kolmogorov evolution describes the redistribution of water among the tanks. Specifically, $Q_{ii}$ is the proportionality constant governing the rate at which water is drained from tank $i$ relative to its water level $p_i$, while $Q_{ij}$, for $i\neq j$, is the rate constant governing the transfer of water from tank $j$ to tank $i$. Thus, the norm constraints~\eqref{eq:norms} impose a limit on the total draining power of the system, while the transfer of water between tanks is cost-free. Moreover, the conservation of total probability corresponds to the conservation of the total amount of water in the system, with no water leaking or spilling. The Kolmogorov dynamics accurately reproduce the evolution of this hydraulic system, as formalized in Appendix~\ref{app: toy}. Here, we provide only a qualitative description of this correspondence. In Fig.~\ref{fig:water}, we display such a system for $N = 3$. In panel $(a)$ we show the analogy between states and water level distributions and in panel $(b)$ we illustrate an instant in the transfer of water from the flat distribution to the one with all water in tank 1. For further details, see Appendix \ref{app: toy}.
\begin{figure}[h!]
    \centering
\includegraphics[width=0.9\linewidth]{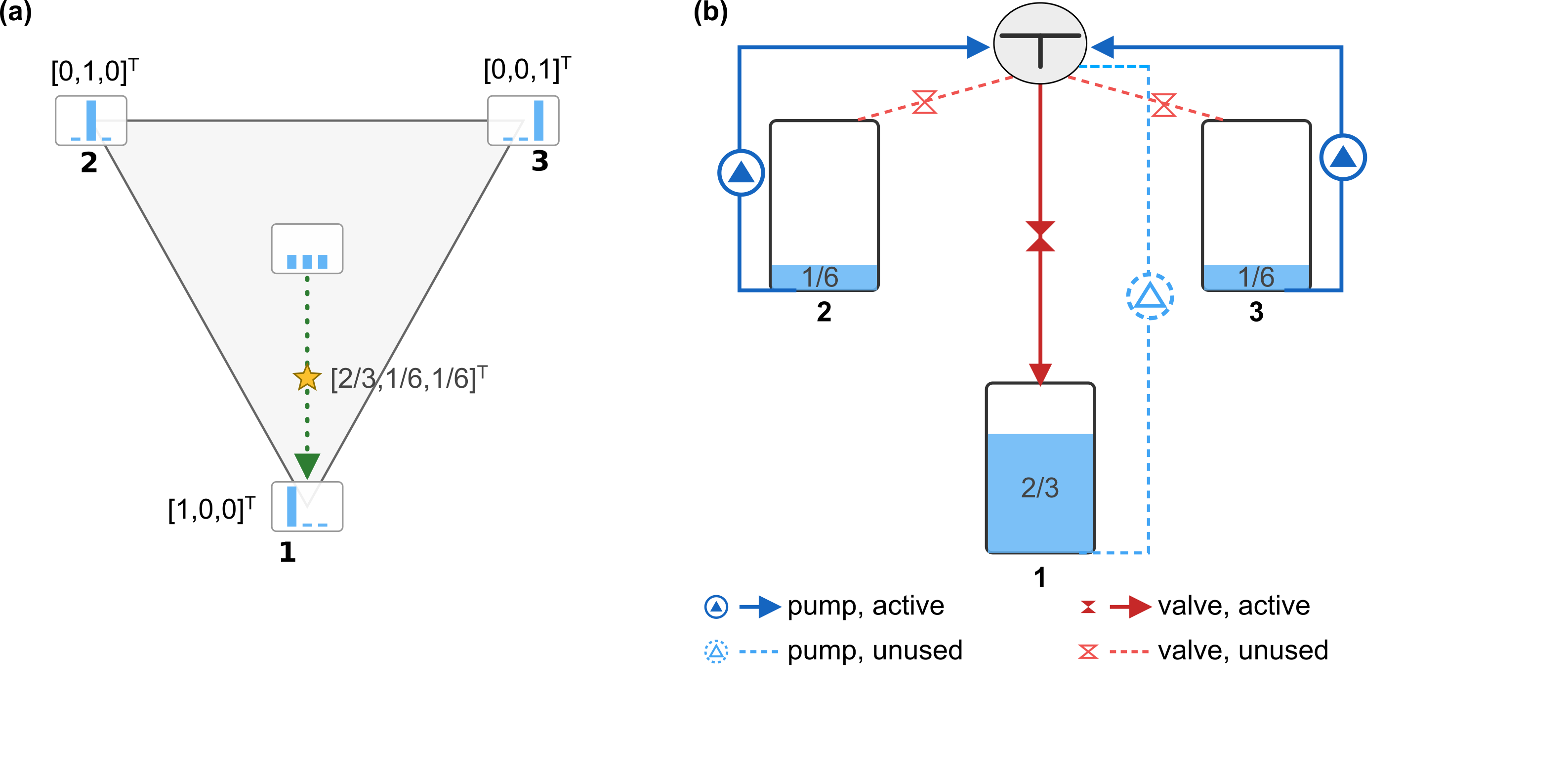}
    \caption{$(a)$ Statics of the water system. Each state in the 2-simplex corresponds to a configuration of water levels in three tanks. $(b)$ Water-system dynamics corresponding to the green dotted arrow at the time marked by the star in $(a)$. Pipe color indicates flow direction in the \textit{medical} convention: blue carries water from a tank to the distributor (T), red returns it from the distributor to a tank; dashed lines mark unused connections. Triangles denote (costly) pumps; double triangles denote valves, usage of which is free.}
    \label{fig:water}
\end{figure}

  
Given an initial state (distribution of water), which other configurations can be reached within a given time under a prescribed constraint on the generator norms? Let us introduce the following definitions, which apply for arbitrary norms, not just isotropic.

\begin{definition}[Outgoing reachable sets]
    The outgoing reachable set $\mc{S}^{+}_t({\bf{p}})$ is defined to be the set of all states that can be reached from an initial ${\bf{p}}$ in time no larger than $t$ via time-inhomogeneous Markov dynamics bounded by a certain norm $\| \cdot \|$
    \begin{equation}
        \mc{S}^{+}_t({\bf{p}})=\{{\bf{q}} \tx{ s.t. }  {\bf{q}} = \mc{T}e^{\int_0^tQ_sds}{\bf{p}} \tx{ with }\|Q_s\|\leq 1 \ \forall s\in[0,t]\}.
    \end{equation}
\end{definition}

Due to the inherent irreversibility of Markov dynamics, the set of states that \textit{can be reached} starting from ${\bf{p}}$ in a time $t$ is different from the set of states that \textit{can reach} ${\bf{p}}$ in the same time. Hence it is natural to also define

\begin{definition}[Incoming reachable sets]
    The incoming reachable set $\mc{S}^{-}_t({\bf{p}})$ is defined as the set of all states that can reach a given target state ${\bf{p}}$ in time no larger than $t$ via time-inhomogeneous Markov dynamics bounded by a given norm $\| \cdot \|$
    \begin{equation}
        \mc{S}^{-}_t({\bf{p}})=\{{\bf{q}} \tx{ s.t. }  {\bf{p}} = \mc{T}e^{\int_0^tQ_sds}{\bf{q}} \tx{ with }\|Q_s\|\leq 1 \ \forall s\in[0,t]\}.
    \end{equation}
\end{definition}

In Fig.~\ref{fig:qutrit_markov}, we display a sketch of the outgoing and incoming reachable sets under the norm $||Q ||_{\infty}^{~\rm{diag}}$ for the qutrit, which we formalize in detail in Section~\ref{sec:results}.   

\begin{figure}[h]
    \centering
    \includegraphics[width=0.97\linewidth]{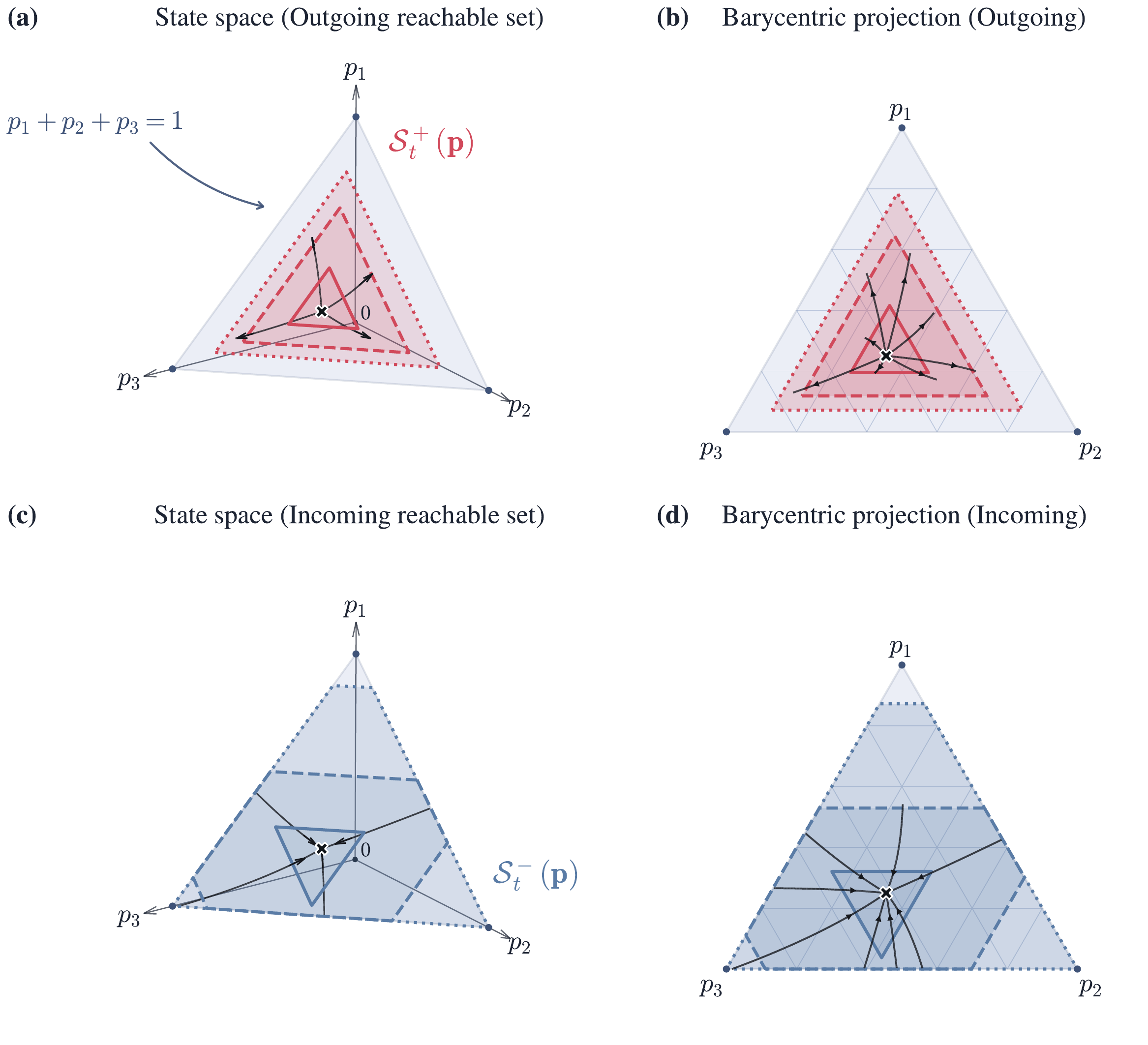}
   \caption{Outgoing and incoming reachable sets for simplex of  $N=3$ probability vectors under the constraint
$\|Q\|_{\infty}^{~\mathrm{diag}}=\frac{1}{N}\max_i|Q_{ii}|\leq 1$ (with time-scale $\tau=1$).
Panels (a,b) show the outgoing sets $S_t^{+}(\mathbf{p})$ and panels (c,d) the incoming sets
$S_t^{-}({\mathbf{p}})$ for $t=1/4,3/4, 5/4$ (solid, dashed, and dotted boundaries,
respectively). The cross marks the initial or target state $\mathbf{p}$, and the black curves
show sample Markovian trajectories.}
    \label{fig:qutrit_markov}
\end{figure}

It should be apparent that the reachable sets under time-homogeneous Markov dynamics are included as subsets of $\mc{S}^\pm_t({\bf{p}})$, as this is just a particular case of time-inhomogeneous evolution. Thus, in the following we will be just concerned with the reachable sets under the latter. Additionally, one can describe the set of tangent vectors (instantaneous velocities) at a given state ${\bf{p}}$ that can be achieved with generators bounded 
with respect to a given norm $\|\cdot\|$.
\begin{definition}
    We denote as $\mc{J}_{\bf{p}}$ the set of tangent vectors that can be achieved via generators bounded by some norm $\|\cdot\|$,
    \begin{equation}
        \mathcal{J}_{\bf{p}} = \{{\bf{v}}\in T_{\bf{p}}\Delta_{N-1} \mathrm{~s.t.~} {\bf{v}} = Q{\bf{p}},~\|Q\|\leq 1\}.
    \end{equation}
    Its boundary -- the set of largest achievable velocities -- is denoted as $\mathcal{I}_{\bf{p}}$.
\end{definition}

Hence, finding the reachable sets induced by a given constraint $\|Q\|\leq1$ becomes the central question, as these inform us of the allowed motions and states that are reachable from an initial one in  a certain time $t$. It should be clear that the shapes of these  velocity sets and reachable sets will depend on the choice of norm $\|\cdot\|$ -- on how the dynamics are constrained -- and, for most choices of norm $\|\cdot\|$, these shapes will be highly complex. Yet, as will be shown in Section~\ref{sec:results}, it is possible to find universal bounds to these sets which apply for all isotropic norms of same time-scale.

\subsection{The time quasi-distance}

\

In this subsection we present the quasi-distance giving the minimum time to join pairs of states when the condition $\|Q\|\leq 1$ is imposed. In particular, we show it to be compatible with a Finsler metric, which allows to cast the finite-time reachability problem in a geometric language.

\begin{definition}[Minimal-time function]
    For a given norm $\|\cdot\|$ on $\mk{gen}$, we define the minimal-time function $T:\Delta_{N-1}\times\Delta_{N-1} \ra [0,\infty]$ of a pair of states $({\bf{p}},{\bf{q}})$ as the smallest time for the first to evolve into the second via time-inhomogeneous Markov dynamics bounded as $\|Q_t\|\leq 1$,
    \begin{equation}\label{eq: Tquasimetrics}
        T({\bf{p}},{\bf{q}}) := \inf_{Q_s} t \tx{ s.t. } {\bf{q}} = \mc{T}e^{\int_0^tQ_sds}{\bf{p}} \tx{ with }\|Q_s\|\leq 1 \ \forall s \in [0,t]
    \end{equation}
\end{definition}

This minimal-time function has some distance-like features. More precisely,

\begin{proposition}\label{th: T quasidistance}
    The function $T(\cdot,\cdot)$ induced by a given norm on generators $\|\cdot\|$ satisfies
    \begin{itemize}
        \item $T({\bf{p}},{\bf{q}}) \geq 0 \ \forall {\bf{p}},{\bf{q}}$.
        
        \item $T({\bf{p}},{\bf{q}}) = 0 \iff {\bf{p}} = {\bf{q}}$.

        \item $T({\bf{p}},{\bf{q}}) \leq T({\bf{p}},{\bf{r}}) + T({\bf{r}},{\bf{q}}) \ \forall {\bf{p}},{\bf{q}},{\bf{r}}$.
    \end{itemize}
    Since, in general, $T({\bf{p}},{\bf{q}}) \neq T({\bf{q}},{\bf{p}})$, this makes $T(\cdot,\cdot)$ a quasi-distance \cite{eckstein2025maxtypequasimetricsprobabilitysimplices}.
\end{proposition}

\noindent For the proof, see Appendix~\ref{app:prelim}.

\

Note that providing an explicit form for the $T(\cdot,\cdot)$ induced by some norm $\|\cdot\|$ solves the finite-time reachability problem, as $\mathcal{S}^{+}_t(\mathbf{p}) = \{\mathbf{q} \mathrm{~s.t.}~T(\mathbf{p},\mathbf{q})\leq t \}$ and $\mathcal{S}^{-}_t(\mathbf{p}) = \{\mathbf{q} \mathrm{~s.t.}~T(\mathbf{q},\mathbf{p})\leq t \}$. However,
the time $T$ in Eq.~\eqref{eq: Tquasimetrics} is challenging to evaluate explicitly, in part because the minimization over generators is a nonconvex problem and, thus, numerical solvers would be prone to getting trapped in local minima. In this work we exploit the geometric properties of the time quasi-distance $T(\cdot,\cdot)$ to provide rigorous analytical bounds to the sets of reachable sets. Concretely, we show that  the quasi-distance $T(\cdot,\cdot)$ can be seen to be induced by a Finsler metric, whose form is determined by the particular norm $\|\cdot\|$ employed to constrain the system.

\begin{theorem}[Finsler time-metric]\label{th: time norm}
    The time quasi-distance $T(\cdot,\cdot)$, as defined in \eqref{eq: Tquasimetrics}, is the geodesic quasi-distance induced by the $C^0$-Finsler metric $F({\bf{p}},{\bf{v}}) = \min_{\mk{gen}({\bf{p}},{\bf{v}})}\|Q\|$:
    \begin{equation}
    \label{eq:min_rollo}
        T({\bf{p}},{\bf{q}}) = \min_{\substack{\gamma \ \tx{s.t.}\\ \gamma(0)={\bf{p}} \\ \gamma(1) = {\bf{q}}}} L[\gamma],
    \end{equation}
    where the minimization is performed over all directed curves joining ${\bf{p}}$ at $s=0$ to ${\bf{q}}$ at $s=1$, and $L[\gamma] = \int_0^1F(\gamma(s),\dot\gamma(s))\, ds$ is its corresponding length functional.

\end{theorem}

\noindent For the proof, see Appendix~\ref{app:prelim}.

\

The characterization of $T$ as the quasi-distance induced by the Finsler metric
\begin{equation}\label{eq: TFinslernorm}
    F({\bf{p}},{\bf{v}})= \min_{\mk{gen}({\bf{p}},{\bf{v}})} \|Q\|,
\end{equation}
offers a more direct relation between the constraint $\|Q\|\leq 1$ and the corresponding ``time-geometry'' of optimal dynamics. Unlike the general optimization \eqref{eq: Tquasimetrics}, for isotropic norms the Finsler metric \eqref{eq: TFinslernorm}
is easier to characterize. Furthermore, its geodesics and metric balls -- that is, its optimal trajectories and reachable sets -- can be explicitly obtained for $\ell_\alpha$-norms on the diagonal of the generators, with the two extremal cases $\alpha=1$ and $\alpha=\infty$ respectively providing inner and outer bounds to the optimal dynamics under any other isotropic norms with same time-scale. These will be examined below.


\section{Main results} \label{sec:results}


\




We start the characterization locally, addressing the Finsler metric and set of achievable velocities induced by dynamics that are bounded by an isotropic norm.





\begin{proposition}[Achievable velocities for arbitrary isotropic norms]
    The Finsler metric due to dynamics bounded in the isotropic norm $\|Q\|:=f(Q_{11},...,Q_{NN})$ is
    \begin{equation}
        F({\bf{p}},{\bf{v}}) = f(|v_{i_1}|/p_{i_1},...,|v_{i_{n^-}}|/p_{i_{n^-}}, 0,...,0) \tx{ for }i_1,...,i_{n^-}\in I_-, 
    \end{equation}
    where $I_-$ indexes the negative components of ${\bf{v}}$. Its unit ball at a fixed ${\bf{p}}$ gives the corresponding set of achievable velocities, $\mc{J}_{\bf{p}}$.
\end{proposition}

\begin{proof}
    Given an isotropic norm $\|Q\|=f(Q_{11},...,Q_{NN})$, the result directly follows from \eqref{eq: TFinslernorm} when taking into account the considerations on minimal generators discussed in Proposition \ref{th: gen(p,v)2}. Also, since $F({\bf{p}},{\bf{v}})= \min_{\mk{gen}({\bf{p}},{\bf{v}})}\|Q\|$, a velocity with $F({\bf{p}},{\bf{v}})>1$ is only obtainable with non-allowed generators ($\|Q\|>1$), and so $F({\bf{p}},{\bf{v}})\leq 1$ determines the achievable velocities.
\end{proof}

Observe that $F({\bf{p}},{\bf{v}})$ strongly depends on the signs of the entries of ${\bf{v}}$. In particular, when $\textbf{v}$ has a single negative component, namely $v_i$, then $F({\bf{p}},{\bf{v}})=f(|v_i|/p_i,0,...,0)$. But by (\ref{eq: time-scale}) the behavior of this quantity does not depend on the precise details of the underlying norm $\|\cdot\|$, just on its time-scale.
\begin{equation}
    F({\bf{p}},{\bf{v}}) = \tau |v_i|/p_i \tx{ whenever } {\bf{v}} \tx{ has \ a \ single \ negative \ entry } v_i.
\end{equation}
That is, the Finsler time metrics $F$ induced by arbitrary isotropic norms $\|\cdot\|$ with the same time-scale $\tau$ all take the same form on the single-negative-entry chambers of $T_{\bf{p}}\Delta_{N-1}$. This fact will allow us to find bounds to the sets of maximum velocities, $\mathcal{J}_{\mathbf p}$, in these contexts.

\begin{proposition}[Inner and outer bounds for arbitrary isotropic norms]\label{th: TS outer}
    The set of achievable velocities due to dynamics bounded in an arbitrary isotropic norm of time-scale $\tau$ is inner and outer bounded by the sets
    \begin{equation}\label{eq: trace_extremals}
        \mc{J}_{\bf{p}}^\tx{inner} = \tx{conv}\{{\bf{v}}_{(ij)}\}_{i\neq j}, \tx{ where }
        {v_{(ij)}}_k=\begin{cases}
            -p_i/\tau &\tx{ if }k=i\\
            p_i/\tau &\tx{ if }k=j\\
            0 &\tx{otherwise}
        \end{cases}
        \!\!\! \tx{ and } \ 
        \mc{J}_{\bf{p}}^\tx{outer} = \frac{1}{\tau}(\Delta_{N-1}-{\bf{p}}).
    \end{equation}
    These are, respectively, the sets of achievable velocities under the norms $\|Q\|_1^\tx{diag}=\tau\sum_i|Q_{ii}|$ and $\|Q\|^\tx{diag}_\infty = \tau\max_i|Q_{ii}|$.
\end{proposition}

\noindent For the proof, see Appendix \ref{pr: TS outer}.

\begin{figure}[h!]
    \centering
    \includegraphics[width=1\linewidth]{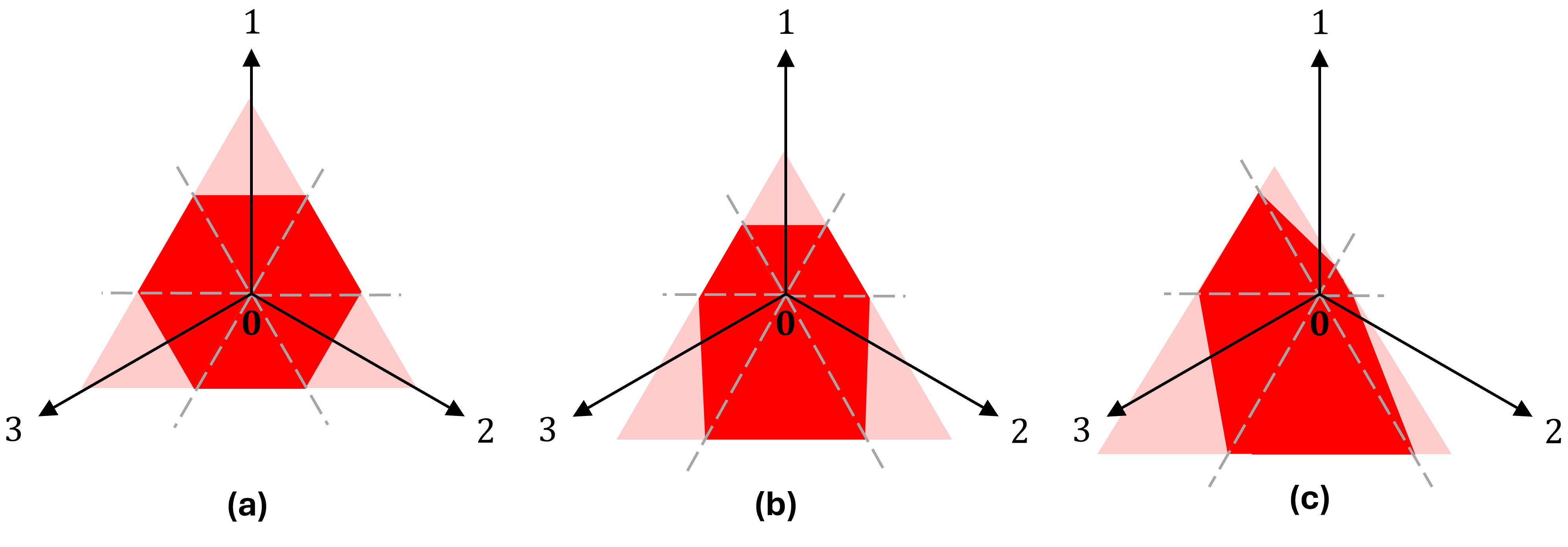}
    \caption{The sets $\mathcal{J}_{\mathbf{p}}^{\text{inner}}$ (bright red) and $\mathcal{J}_{\mathbf{p}}^{\text{outer}}$ (light red) for a three-level system. Panel (a) shows the sets at the maximally mixed initial state $\mathbf{p} = [1/3, 1/3, 1/3]^T$. Panels (b) and (c) display these sets at initial states with increasing purity, specifically $\mathbf{p} = [0.4, 0.3, 0.3]^T$ and $\mathbf{p} = [0.55, 0.35, 0.1]^T$, respectively.}
    \label{fig:TraceNormIndicatrices}
\end{figure}
In Fig.~\ref{fig:TraceNormIndicatrices}, we display the polytope bounds specified in~\eqref{eq: trace_extremals} to the achievable velocities under isotropic norms of a three-level system.

\subsection{Finite-time reachable sets. Inner and outer bounds}

The following Theorem and its Corollary constitute one of the main results of this work.


\begin{theorem}[Time quasi-distances for diagonal $\ell_\alpha$-norms]\label{th: FT inner}
    For any two states $\mathbf{p}, \mathbf{q}\in \Delta_{N-1}$, their minimum reachability time under time-inhomogeneous Markovian dynamics with generators bounded as $\|Q\|^\tx{diag}_\alpha=\tau\left(\sum_i|Q_{ii}|^\alpha\right)^{1/\alpha} \leq 1$ is
    \begin{equation}\label{eq: InnerFinslergen}
         T_\alpha({\bf{p}},{\bf{q}}):= \tau \left( \sum_{i=1}^N \max\left(0,\ln\frac{p_i}{q_i}\right)^\alpha\right)^{1/\alpha}.
    \end{equation}
    A unit-norm time-homogeneous generator joining ${\bf{p}}$ to ${\bf{q}}$ interior in the minimum time $T_\alpha({\bf{p}},{\bf{q}})$ is
    \begin{equation}
        Q_{ij}=
        \begin{cases}
            -\delta_{ij}\ln(p_i/q_i)/T_\alpha({\bf{p}},{\bf{q}}) & \tx{ if } i,j\in I_-\\
            W_{ij}\ln(p_j/q_j)/T_\alpha({\bf{p}},{\bf{q}}) & \tx{ if } i\in I_+,j\in I_-\\
            0 & \tx{ otherwise}
        \end{cases}
    \end{equation}
    where $W_{ij}$ is any entrywise non-negative matrix such that $\sum_{j\in I_-}W_{ij}(p_j-q_j) = q_i-p_i \ \forall i\in I_+$ and $\sum_{i\in I_+}W_{ij} = 1$.
\end{theorem}

\noindent For the proof, see Appendix~\ref{app:results}.

\

Remarkably, time-homogeneous generators saturate the minimum reachability times, even when general time-inhomogeneous dynamics are considered. Also, reachable sets corresponding to situations with $\|Q\|_\alpha^\tx{diag}\leq 1$ follow directly from the explicit characterization of the time quasi-distances given above,
\begin{equation}
    \mc{S}^+_t({\bf{p}}) := \{{\bf{q}} \tx{ s.t. } T_\alpha({\bf{p}},{\bf{q}})\leq t\} \tx{ and } \mc{S}^-_t({\bf{p}}) := \{{\bf{q}} \tx{ s.t. } T_\alpha({\bf{q}},{\bf{p}})\leq t\},
\end{equation}
and, since $T_\infty({\bf{p}},{\bf{q}})\leq T_\alpha({\bf{p}},{\bf{q}})\leq T_{\alpha'}({\bf{p}},{\bf{q}})\leq T_1({\bf{p}},{\bf{q}})$ whenever $1\leq \alpha'\leq \alpha\leq \infty$, their corresponding reachable sets are nested. Furthermore, the reachable sets in the extremal cases $\alpha=1$ and $\alpha=\infty$ constitute inner and outer bounds also to the reachable sets due to dynamics bounded in arbitrary isotropic norms of same time-scale, not just of the diagonal-$\ell_\alpha$ type, as the following Corollary establishes.

\begin{corollary}[Inner and outer bounds for arbitrary isotropic norms]\label{th: FT outer}
    Consider the quasi-distances $T_1({\bf{p}},{\bf{q}}) = \tau \sum_{i=1}^N\max\left(0,\ln\frac{p_i}{q_i}\right) \tx{ and } T_\infty({\bf{p}},{\bf{q}}) = \tau D_{\infty}(\mathbf{p}|\mathbf{q}):= \tau\max_{i} \ln \frac{p_i}{q_i}$, where $D_{\infty}$ denotes the $\infty$-Rényi divergence~\cite{Renyi1961}. For any two states ${\bf{p}}, {\bf{q}}\in \Delta_{N-1}$,
     \begin{equation}
         T_\infty({\bf{p}},{\bf{q}}) \leq T({\bf{p}},{\bf{q}}) \leq T_1({\bf{p}},{\bf{q}}),
     \end{equation}
     where $T$ is the minimal time function due to dynamics bounded by any isotropic norms of time-scale $\tau$. Thus, the reachable sets under the latter, $\mc{S}_t^\pm({\bf{p}})$, satisfy
    \begin{equation}
    \label{eq:sets_outer}
        \mc{S}_t^{1\pm}({\bf{p}}) \subseteq \mc{S}_t^\pm({\bf{p}}) \subseteq\mc{S}_t^{\infty\pm}({\bf{p}}) = (e^{\mp t/\tau}{\bf{p}} + (1-e^{\mp t/\tau})\Delta_{N-1})\cap \Delta_{N-1},
    \end{equation}
    where $\mc{S}_t^{1\pm}({\bf{p}})$ and $\mc{S}_t^{\infty\pm}({\bf{p}})$ are, respectively, the outgoing (incoming) metric balls of $T_1(\cdot,\cdot)$ and $T_\infty(\cdot,\cdot)$.
\end{corollary}

\noindent For the proof, see Appendix \ref{pr: FT outer}. 

\

The formal connection between the time quasi-distance and the Kullback-Leibler relative entropy~\cite{KullbackLeibler1951} is explored further in Section~\ref{sec: KL asymmetry}. Theorem~\ref{th: FT inner} and its Corollary can be considered as complementary to earlier results of Gu ~\cite{Gu2023SpeedLimit}, where the minimal time $T_{\infty}$ determines speed limits under time-independent generators bounded in $\|Q\|_1^\tx{diag\!}$. While this conclusion is consistent with \eqref{eq:sets_outer}, the exact reachability sets in this bounded trace scenario would be given by the tighter $\mc{S}_1^\tx{diag}$. Moreover, our result establishes that the proposed $\mc{S}^{\infty\pm}_t$ speed limit holds for arbitrary isotropic norms with fixed time-scale even if the generator depends explicitly on time, and is exactly saturated if the infinite norm $\|Q\|_{\infty}^\tx{diag}$ is fixed. \\

\begin{figure}[h!]
    \centering
    \includegraphics[width=\linewidth]{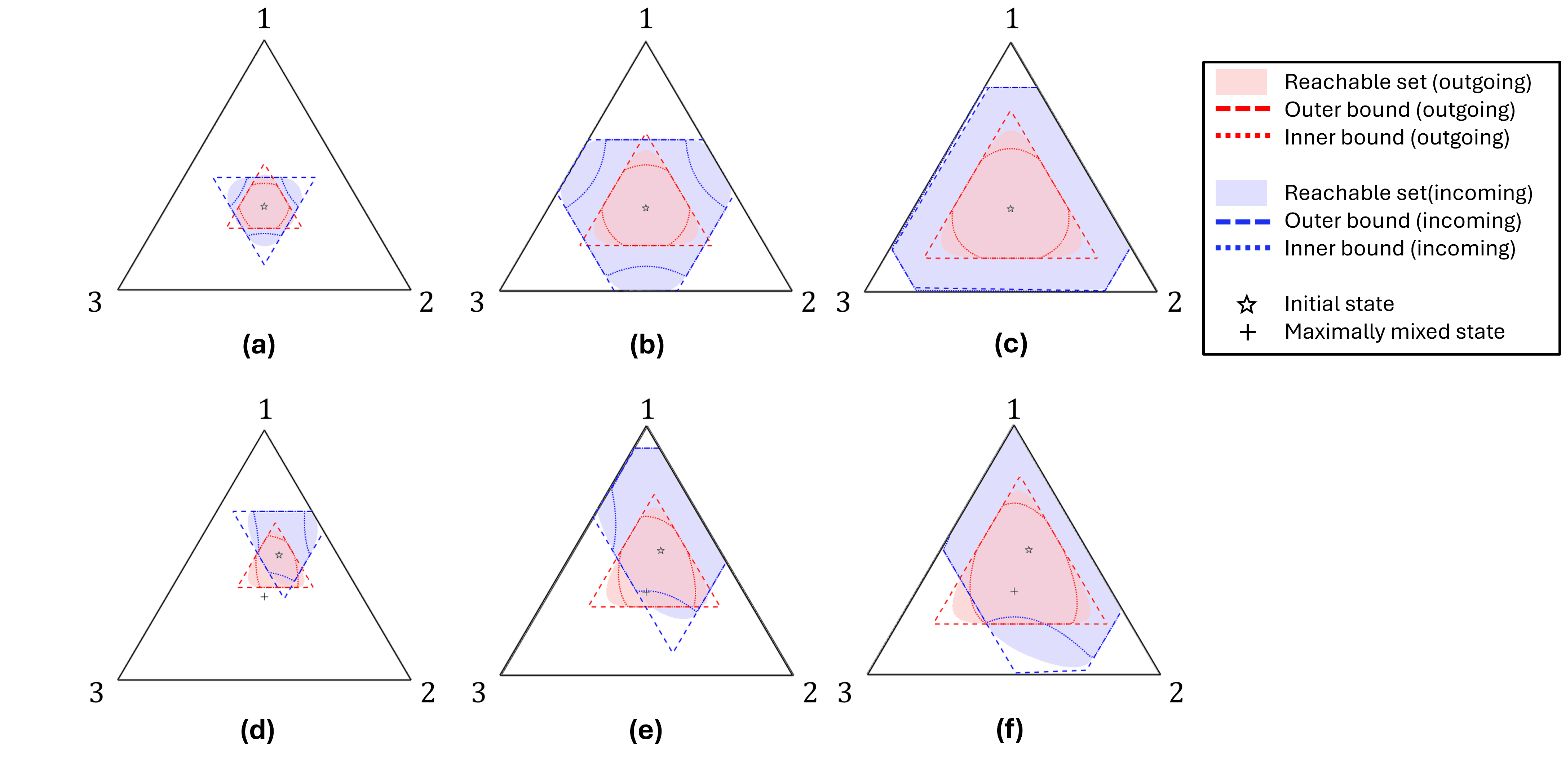}
    \caption{Comparison of the incoming and outgoing reachable sets for $N=3$ when the dynamics is constrained by $\|Q\|_{2}^\tx{diag}:=\sqrt{\sum_{i}Q_{ii}^2}\leq 1$ (with $\tau = 1$) for different times: $t=0.05$ in (a) and (d), $t=0.1$ in (b) and (e), and $t=0.15$ in (c) and (f). As initial states we consider ${\bf{p}}=[1/3,1/3,1/3]^T$ (top row) and ${\bf{p}}=[0.2,0.3,0.5]^T$ (bottom row)}
    \label{c}
\end{figure}

In Fig.~\ref{c} we display the outgoing (red) and incoming (blue) sets with their respective bounds for different initial states. We verify how they differ due to the irreversibility of Kolmogorov evolution. 
It is harder to evolve a probability vector $\bf{p}$ towards a deterministic state (vertices of the simplex) than towards the maximally random state $\mathbf{p} = \omega \equiv (1/N, \dots, 1/N)^T$. This makes the shapes of the outgoing and incoming sets more distinct for initial states near the boundary. The difference in the structure of the incoming and outgoing set for a fixed hiking time in the mountains depends on the difference of the altitude in the terrain (i.e. the gravitational potential), see Fig.
\ref{fig:riemann_finsler}b.
In the studied case of Kolmogorov dynamics in the probability simplex the role of the potential is played by the closest distance of a given probability vector $q$ to the boundary of the simplex, which is determined by its
R{\'e}nyi entropy $S_\alpha(q)$ in the limit 
 the parameter $\alpha$ goes to zero
 \cite{Bengtsson2006}.



\begin{figure}[h!]
    \centering
    \includegraphics[width=0.95\linewidth]{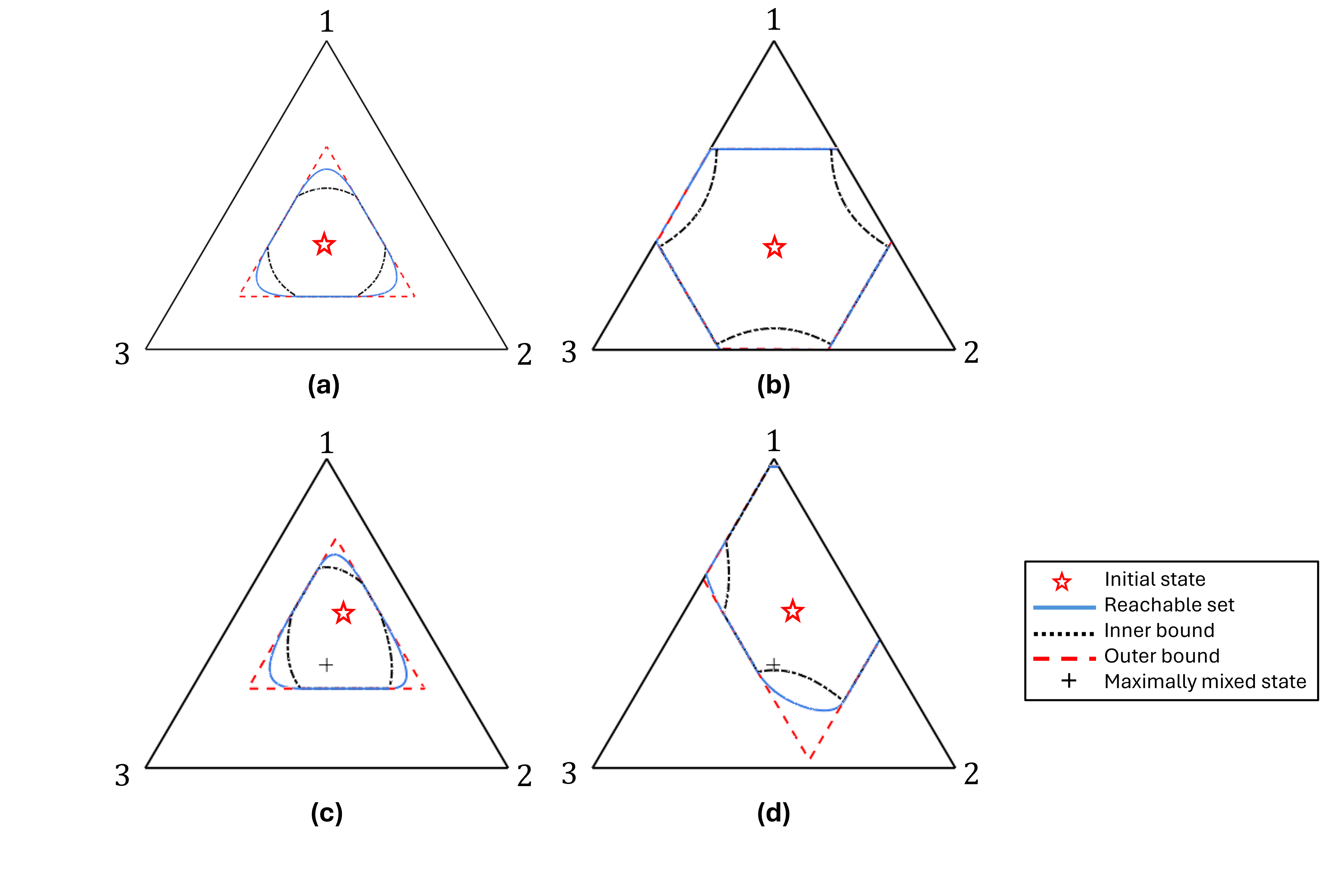}
    \caption{Comparison of the reachable sets (in blue) for $N=3$ under $\|Q\|^\tx{diag}_2:=\sqrt{\sum_iQ_{ii}^2}\leq 1$ in $t=1/9$ (the time-scale is set to $\tau=1$), with the derived bounds in the outgoing (a) and (c), and incoming (b) and (d) cases; with initial states ${\mathbf{p}}=[\frac{1}{3},\frac{1}{3},\frac{1}{3}]^T$ (top row) and $[0.2,0.3,0.5]^T$ (bottom row).}
    \label{fig:timeballsgen}
\end{figure}
In Fig.~\ref{fig:timeballsgen}, we assess the quality of the bounds for different norms by comparing it with the reachable set obtained numerically. It can be seen that the norm $\|Q\|_\tx{1}^{\rm diag}$ is tight for the inner bound, while $\|Q\|_\tx{\infty}^{\rm diag}$ is tightest for the outer bound. In the next section, we further address the tightness and asymptotic behavior of this characterization. \\

\newpage

\subsection{Volumes of reachable sets for arbitrary system dimension}
\label{sec: volumes}

Consider the reachable sets $\mc{S}^\pm_t({\textbf{p}})$ induced by certain isotropic norm $\|\cdot\|$ of time-scale $\tau$. It is pertinent to ask how tight are the corresponding inner and outer bounds, found in Theorem \ref{th: FT inner} and Corollary \ref{th: FT outer}. Moreover, it is legitimate to inquire whether the quality of the bounds increases or decreases in time, and how does it depend on the system dimension $N$. For simplicity, we will just analyze the outgoing case when the initial state is the maximally mixed, i.e. how tight are bounds $\mc{S}^{\infty+}_t(\omega)$ and $\mc{S}^{1+}_t(\omega)$ with respect to $\mc{S}^+_t(\omega)$ induced by some $\|\cdot\|$ of time-scale $\tau$. Reasonable quantifiers are, respectively, the ratios
\begin{equation}
    \varepsilon^\tx{outer}
    _t\!\!:=\frac{\tx{vol}[\mc{S}^+_t(\omega)]}{\tx{vol}[\mc{S}^{\infty+}_t\!(\omega)]} \ \tx{ and } \ \varepsilon^\tx{inner}
    _t\!\!:=\frac{\tx{vol}[\mc{S}^{1+}_t\!(\omega)]}{\tx{vol}[\mc{S}^+_t(\omega)]}
\end{equation}
and the gap estimators
\begin{equation}
    \delta^\tx{outer}_t\!\!:=\frac{\tx{vol}[\mc{S}^{\infty+}_t\!(\omega)]-\tx{vol}[\mc{S}^+_t(\omega)]}{\tx{vol}[\Delta_{N-1}]} \ \tx{ and } \ \delta^\tx{inner}_t\!\!:=\frac{\tx{vol}[\mc{S}_t^+(\omega)]-\tx{vol}[\mc{S}^{1+}_t\!(\omega)]}{\tx{vol}[\Delta_{N-1}]}
\end{equation}
where $\tx{vol( \cdot )}$ denotes volume computed with respect to the flat measure on the simplex. Furthermore, since under $\|Q\|\leq 1$ for any isotropic norm of time scale $\tau$ one has $\tx{vol}[\mc{S}_t^{1+}\!(\omega)] \leq \tx{vol}[\mc{S}_t^\tx{+}\!(\omega)] \leq \tx{vol}[\mc{S}_t^{\infty +}\!(\omega)]$, we may just work with the norm-independent quantities
\begin{align}
    \varepsilon_t&:=\frac{\tx{vol}[\mc{S}^{1+}_t\!(\omega)]}{\tx{vol}[\mc{S}^{\infty+}_t\!(\omega)]} \leq \varepsilon^\tx{inner}\!\!,\varepsilon^\tx{outer}\\
    \delta_t&:=\frac{\tx{vol}[\mc{S}^{\infty+}_t\!(\omega)]-\tx{vol}[\mc{S}^{1+}_t\!(\omega)]}{\tx{vol}[\Delta_{N-1}]} \geq \delta^\tx{inner}_t,\delta^\tx{outer}_t
\end{align}
which then give worst-possible bounds to the tightness estimators between the reachable set due to any isotropic norm and the $\mc{S}_t^{1,\infty}$ of same timescale $\tau$. Apart from their universal applicability, the advantage of these parameters is that they just rely on the inner and outer sets, already known. The outer volume is
\begin{equation}\label{eq: OuterVol}
    V^+_\tx{outer}(t) := \tx{vol}[\mc{S}^{\infty+}_t\!(\omega)] = \left(1-e^{-t/\tau}\right)^{N-1}V_\tx{total},    
\end{equation}
where $V_\tx{total} = \tx{vol}[\Delta_{N-1}]$ is the volume of the simplex. On the other hand, the inner volume $V^+_\tx{inner}(t):=\tx{vol}[\mc{S}^{1+}_t(\omega)]$ has no elementary closed form, and we estimate it by Monte Carlo integration.\footnote{Target states are drawn uniformly from the simplex,
and accepted when $T^{\tx{inner}}(\omega,\mathbf q)\leq t$. The accepted fraction estimates $V_{\tx{inner}}^+(t)/V_{\tx{total}}$.}

To assess how the tightness of the bounds depends on the system dimension, we plot $\varepsilon$ evaluated at fixed times, against $N$; but to compare systems of different dimension on a common time, we must first specify the dependence of the time-scale $\tau$ on $N$. The next Theorem suggests that a natural choice is $\tau=1/N$ since, with it, a characteristic time $t_c$ emerges which, is independent of the system size, and in the large $N$ limit determines whether the bounds can be considered tight nor not.

\begin{theorem}
\label{prop: threshold}
    Consider the average of $T_1(\omega,{\bf{q}})$ computed with respect to the flat measure on the simplex, $d\mu({\bf{q}})$. With $\tau=1/N$, the following limit is well defined,
    \begin{equation}
    \label{eq: threshold}
        t_c:=\lim_{N\ra\infty} \frac{1}{V_{\!\tx{total}}}\int_{\Delta_{N-1}}\!\!\!\!\!T^\tx{\!inner}\!(\omega,{\bf{q}}) \, d\mu({\bf{q}}) = \gamma + E_1(1) \approx 0.7966,
    \end{equation}
    where $\gamma\approx 0.5772$ is the Euler--Mascheroni constant and $\mathrm E_1(1)=\int_1^\infty e^{-x}/x\,dx \approx 0.2194$. Then, for every fixed $t>0$ with $t\neq t_c$,
    \begin{equation}
    \lim_{N\to\infty}\varepsilon_N(t)
        =
        \begin{cases}
            0 & \tx{if} \ t<t_c,\\[1mm]
            1 & \tx{if} \ t>t_c.
        \end{cases}
        \label{eq: tightness-threshold}
    \end{equation}
\end{theorem}

\noindent The proof is given in Appendix~\ref{app:results}.

\

In Fig.~\ref{fig:tightness_vs_N}, we display the relative tightness
$\varepsilon_t(N)$ now as a function of the system dimension $N$ for
several fixed times chosen above and below the found threshold $t_c$. From here on, we set to the common time scale $\tau=1/N$.

\begin{figure}[H]
    \centering
    \includegraphics[width=0.95\linewidth]
    {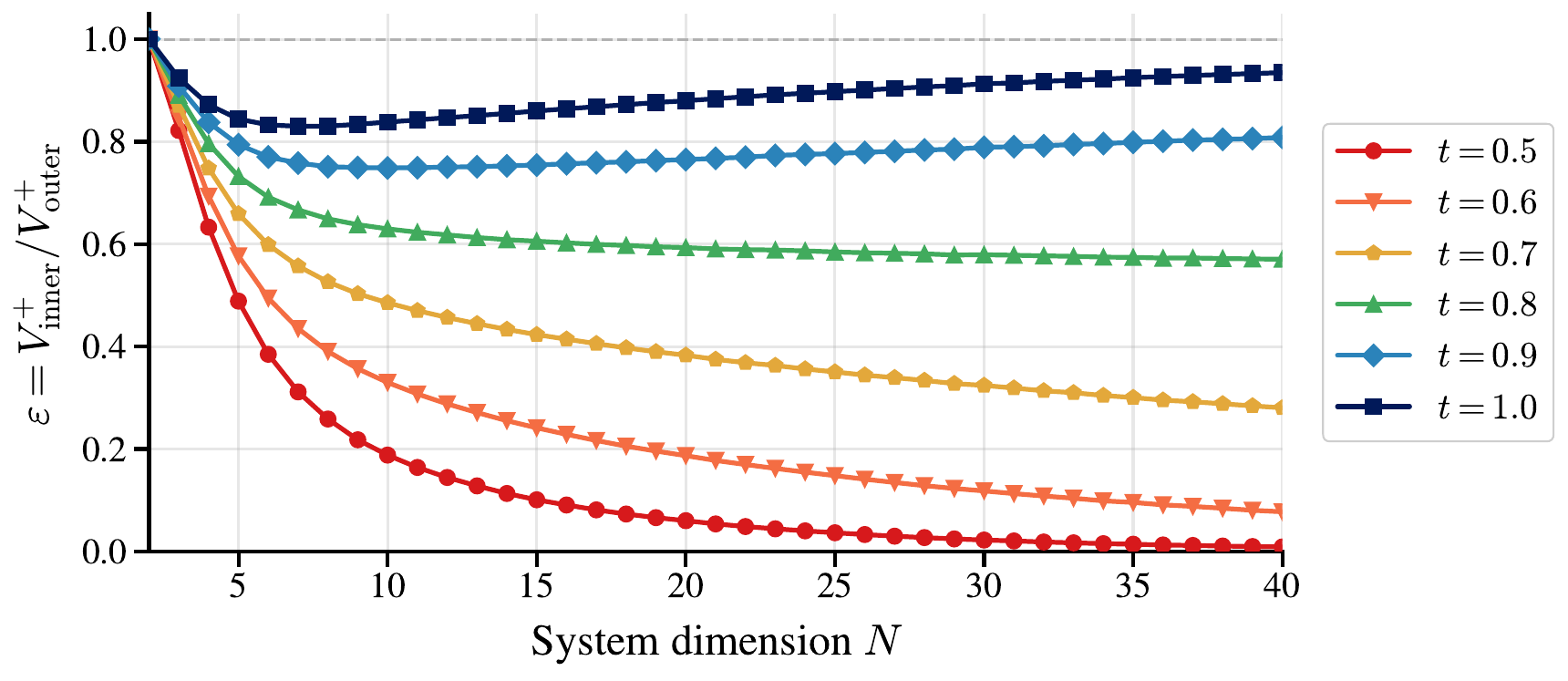}
    \caption{Relative tightness
    $\varepsilon_N(t)=V_{\mathrm{inner}}^{+}(t)/V_{\mathrm{outer}}^{+}(t)$
    of the outgoing-set bounds from
    $\omega$ as a function of $N$, with
    $\tau=1/N$. The inner volumes are estimated using $1.5\times10^{6}$
    uniformly sampled points per dimension, while
    $V_{\mathrm{outer}}^{+}(t)/V_{\mathrm{total}}
    =(1-e^{-Nt})^{N-1}$ is evaluated analytically.}
    \label{fig:tightness_vs_N}
\end{figure}

Thus, we find the behavior anticipated in Theorem~\ref{prop: threshold}. When $\varepsilon$ evaluated at $t<t_c$ is plotted against $N$, we observe a loss of tightness in the large-$N$ limit, and the contrary when $t>t_c$. Since for large systems, $\tau=1/N\ll t_c$, we conclude that tight bounds are achieved at a time much larger that their characteristic timescale (in particular, $t\sim \mc{O}(N\tau)$).

The complementary plots -- the time dependence of the volumes for several fixed values of $N$ -- are shown in Fig.~\ref{fig:tightness}. The two panels quantify
different aspects of the accuracy of the bounds. Panel~\textbf{(a)}
shows the normalized absolute gap $\delta_N(t)$. At very short times this gap is small because both the inner and outer sets occupy only a small fraction of the simplex. It becomes largest in an intermediate time window, when the outer set has already expanded substantially but the inner set has not yet reached the bulk of the simplex. At later times the inner volume catches up with the outer one, and
$\delta_N(t)\to0$.

Panel~\textbf{(b)} shows $\varepsilon_N(t)$ for each finite $N$.
The implications of Theorem~\ref{prop: threshold} are appreciable in how for larger choices of $N$, the curve $\varepsilon_N(t)$ approaches a step function over $t_c$.

\begin{figure}[H]
    \centering
    \includegraphics[width=\linewidth]
    {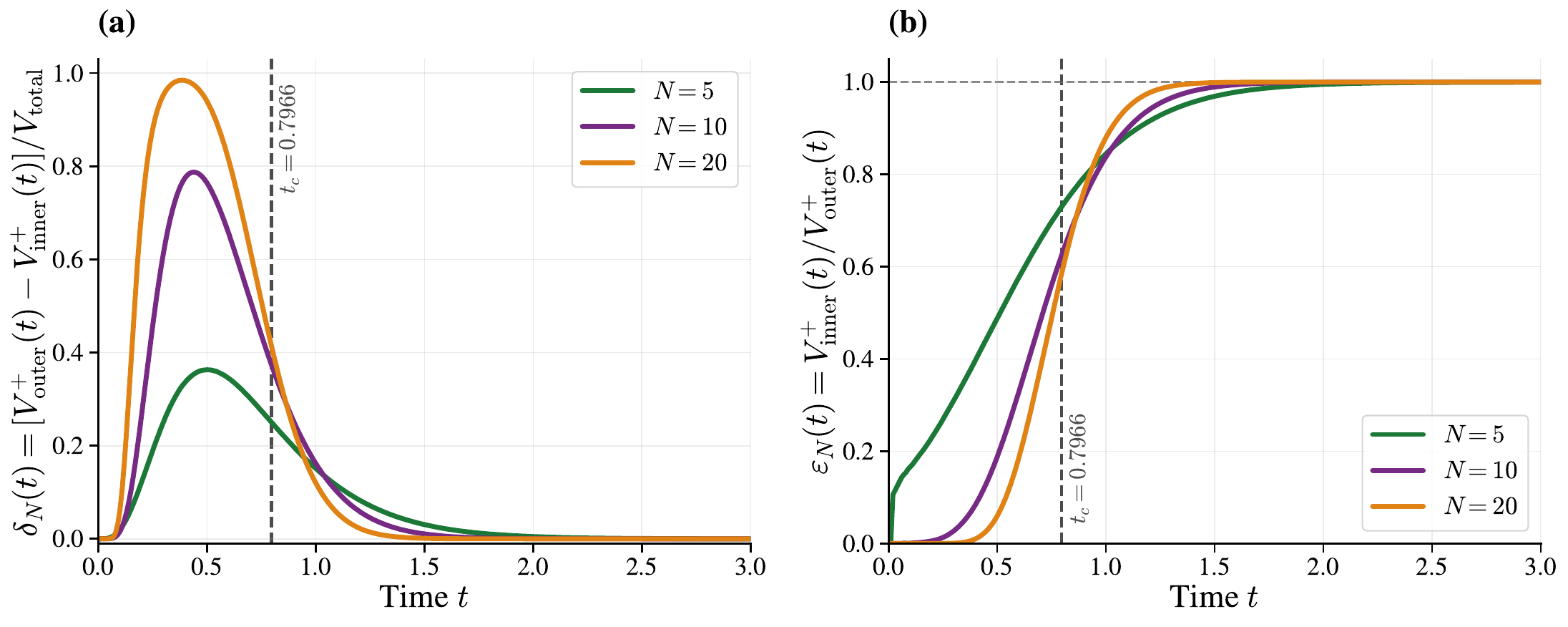}
    \caption{Tightness of the outgoing volume bounds from $\omega$ for
    $N=5,10,20$ and $\tau=1/N$. Panel ({\bf a}) shows the normalized gap
    $\delta_N(t)=[V_{\mathrm{outer}}^{+}(t)-V_{\mathrm{inner}}^{+}(t)]/
    V_{\mathrm{total}}$. Panel ({\bf b}) demonstrates that, for large dimension $N$, the relative tightness
    $\varepsilon_N(t)=V_{\mathrm{inner}}^{+}(t)/V_{\mathrm{outer}}^{+}(t)$ approaches step-function behavior at $t_c$.}
    \label{fig:tightness}
\end{figure}

\subsection{Asymmetry of reachable sets and relative entropy}
\label{sec: KL asymmetry}

The difference between the incoming and outgoing sets, or the asymmetry of the induced time quasi-distances provide a signature of the irreversibility of Kolmogorov dynamics.
Throughout this subsection, we retain the normalization $\tau=1/N$
and specialize to the norm $\|\cdot\|_1^\tx{diag}$, whose associated reachable sets and quasi-distance are exactly the ones described in Theorem~\ref{th: FT inner}. For simplicity, we write
$T \equiv T_{1}$ and $\mc{S}^{1\pm}_t\equiv\mc{S}^\pm_t$. In this setting, we can find a quantitative criterion on pairs of states that determines the asymmetry in the time-optimal dynamics joining them.

\begin{theorem}
\label{prop: antisymmetric-time-cost}
For any two interior states
$\mathbf p,\mathbf q\in\Delta_{N-1}$,
\begin{equation}
\begin{aligned}
    T(\mathbf p,\mathbf q)-T(\mathbf q,\mathbf p)
    =
    \frac{1}{N}\sum_{i=1}^{N}\ln\frac{p_i}{q_i}
    =
    \frac{1}{N}
    \ln\frac{\prod_{i=1}^{N}p_i}
             {\prod_{i=1}^{N}q_i}.
\end{aligned}
\label{eq: general-time-asymmetry}
\end{equation}
Consequently,
\begin{equation}
    T(\mathbf p,\mathbf q)\geq T(\mathbf q,\mathbf p)
    \quad\Longleftrightarrow\quad
    \prod_{i=1}^{N}p_i
    \geq
    \prod_{i=1}^{N}q_i.
    \label{eq: product-order}
\end{equation}
\end{theorem}

\begin{proof}
Equation~\eqref{eq: general-time-asymmetry} follows by direct
calculation from Theorem~\ref{th: FT inner}. On the other hand, the
equivalence~\eqref{eq: product-order} follows from the monotonicity of
the logarithm.
\end{proof}

Equation~\eqref{eq: product-order} compares the two directed times
through the coordinate product, so it does not define a total mixedness order on arbitrary pairs of probability vectors.
But for transitions involving the maximally mixed state $\omega$, the
difference takes precisely the form of the Kullback--Leibler divergence
\cite{KullbackLeibler1951}. First, note that

%
\begin{equation}
\begin{aligned}
    T(\omega,\mathbf q)
    =
    \frac{1}{N}\sum_{i:\,q_i<1/N}
    \bigl[-\ln(Nq_i)\bigr], \quad \text{and} \quad
    T(\mathbf q,\omega)
    =
    \frac{1}{N}\sum_{i:\,q_i>1/N}
    \ln(Nq_i).
\end{aligned}
\label{eq: outgoing-incoming-costs}
\end{equation}
The first expression is the minimal time required to prepare
$\mathbf q$ from $\omega$, whereas the second is the minimal time
required to relax $\mathbf q$ to $\omega$. Then,

\begin{proposition}
\label{prop: KL gap}
For every interior state
$\mathbf q\in\operatorname{int}\Delta_{N-1}$,
\begin{equation}
    T(\omega,\mathbf q)-T(\mathbf q,\omega)
    =
    D_{\tx{KL}}(\omega\,\|\,\mathbf q)
    \geq0,
    \label{eq: KL identity}
\end{equation}
where
$D_{\tx{KL}}(\omega\|\mathbf q)
=\sum_i\omega_i\ln(\omega_i/q_i)$ is the
Kullback--Leibler divergence~\cite{KullbackLeibler1951}.
Equality holds if and only if $\mathbf q=\omega$.
\end{proposition}

\begin{proof}
Setting $\mathbf p=\omega$ in~\eqref{eq: general-time-asymmetry} yields
\begin{equation}
\begin{aligned}
    T(\omega,\mathbf q)-T(\mathbf q,\omega)
    =
    \frac{1}{N}\sum_{i=1}^{N}
    \ln\frac{1/N}{q_i}
    =
    D_{\tx{KL}}(\omega\,\|\,\mathbf q).
\end{aligned}
\label{eq: KL-gap-proof}
\end{equation}
Non-negativity and the equality condition follow from Gibbs'
inequality~\cite{Cover2006}.
\end{proof}

\begin{corollary}
If $\mc{S}^\pm_t(\omega)$ denote the incoming and outgoing sets induced by $\|Q\|_1^\tx{diag}\leq 1$, then
\begin{equation}
    \mathcal S_t^+(\omega)\subseteq\mathcal S_t^-(\omega) \ \forall t>0.
    \label{eq:containment}
\end{equation}
\end{corollary}

\begin{proof}
Let $\mathbf q\in\mathcal S_t^+(\omega)$. The first formula in~\eqref{eq: outgoing-incoming-costs} shows that $\mathbf q$ is
interior whenever $t<\infty$. Equation~\eqref{eq: KL identity} gives
\begin{equation}
    T(\mathbf q,\omega)
    \leq T(\omega,\mathbf q)
    \leq t.
\end{equation}
Hence $\mathbf q\in\mathcal S_t^-(\omega)$ by the definitions
above.
\end{proof}

Proposition~\ref{prop: KL gap} states that taking any interior state
$\mathbf q$ to the uniform state $\omega$ is never slower than
preparing it from $\omega$, and that the excess time is exactly
$D_{\tx{KL}}(\omega\|\mathbf q)$. The two costs differ most strongly
near the boundary: $T(\mathbf q,\omega)$ is bounded on
$\Delta_{N-1}$, as shown in Proposition~\ref{prop: saturation}, whereas
$T(\omega,\mathbf q)$ diverges when any component of $\mathbf q$
vanishes, and thus a boundary state can relax to $\omega$ in finite time but cannot be prepared exactly from $\omega$ in finite time. While the conclusion that a boundary state cannot be prepared from an interior one in a finite time holds true for arbitrary bounded generators, the exact identity~\eqref{eq: KL identity} is specific to the dynamics induced by $\|Q\|^\tx{diag}_1\leq 1$; no corresponding equality is asserted here for a general isotropic norm.
Figure~\ref{fig:ratio_over_time} illustrates this asymmetry. For every
$t>0$,~\eqref{eq:containment} gives $V^-(t)/V^+(t)\geq1$. The ratio
tends to one as $t\rightarrow0$, develops a pronounced excess above one
at intermediate times, and returns to one as both volumes approach
$V_{\mathrm{total}}$; over the plotted intermediate-time window this
excess increases strongly with $N$. The incoming set covers the simplex
after the finite time established in
Proposition~\ref{prop: saturation}, whereas the outgoing set covers it
only in the limit $t\to\infty$.

\begin{figure}[h!]
    \centering
    \includegraphics[width=\linewidth]
    {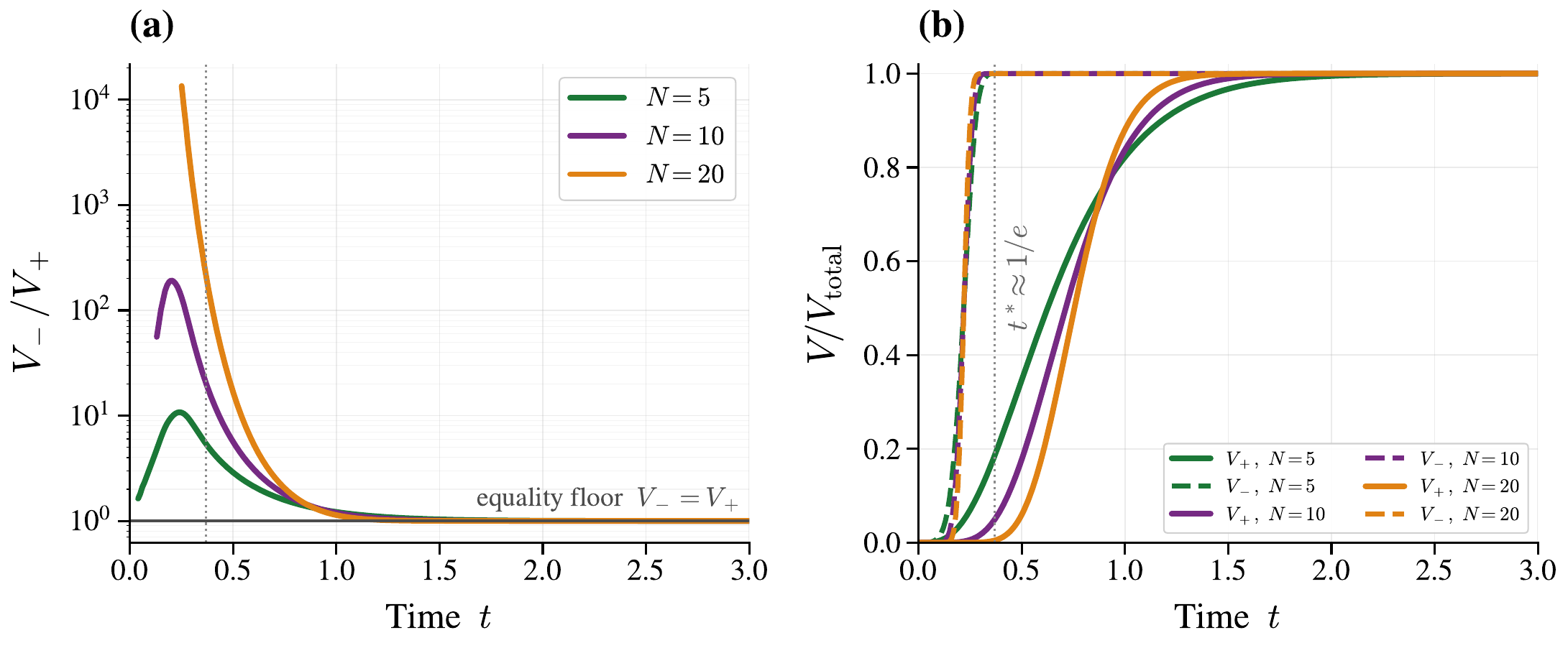}
    \caption{Incoming and outgoing reachable volumes from $\omega$
    under the diagonal-$\ell_1$ constraint
    $\|Q\|_1^{\mathrm{diag}}=|\operatorname{Tr}Q|/N\leq1$ for
    $N=5,10,20$, estimated using $6\times10^{6}$ samples per dimension.
    Panel ({\bf a}) shows $V^-(t)/V^+(t)$ on a logarithmic vertical
    scale. Panel ({\bf b}) shows $V^+(t)/V_{\mathrm{total}}$ (solid
    curves) and $V^-(t)/V_{\mathrm{total}}$ (dashed curves). The vertical
    dotted line at $t=1/e$ is a universal upper bound on the finite-$N$
    incoming covering times and their large-$N$ limit, not the exact
    finite-$N$ value.}
    \label{fig:ratio_over_time}
\end{figure}

The boundedness of $T(\mathbf q,\omega)$ has a direct dynamical
meaning: under the normalization $\tau=1/N$, the worst-case time
required to relax an arbitrary state to the maximally mixed state
remains uniformly bounded as the system dimension grows. This is in
sharp contrast with the outgoing direction, in which states on the
boundary cannot be prepared exactly in finite time. It is easy to show (see Appendix~\eqref{app:results}) that this covering time, after which $\mathcal S_t^-(\omega)=\Delta_{N-1}$, is
\begin{equation}
t_N^*=\max_{1\leq k<N}\frac{k}{N}\ln\frac Nk \longrightarrow \frac{1}{e} \tx{ as \ }N\ra\infty.
\end{equation}

\section{Concluding remarks}
\label{sec:conclusions}

In this work, we take a geometric approach to study the finite-time reachability problem in the probability simplex. In this framework, reachable sets and optimal trajectories correspond, respectively, to the metric balls and geodesic curves of a quasi-distance giving the minimum time to join pairs of states. This quasi-distance has a Finsler-geometric origin, and we show it to be completely determined by the particular norm that bounds the generators. 
In particular, we explicitly characterize the quasi-distances arising when the allowed dynamics have bounded average and maximum collective decay rates, and show that the corresponding reachable sets constitute, respectively, inner and outer bounds to the reachable sets due to dynamics with generators bounded in any other isotropic norm. We also present rigorous approximations of the volumes of such finite-time reachable sets, which confirm the tightness of the proposed bounds after a dimension-independent critical time.  Moreover, some of the discussed quasi-distances can be related to entropic quantities, which give novel insight on the entropy production limits in relaxation processes. In the case of dynamics bounded in the diagonal-$\ell_1$ norm, the difference between the minimum time to go from an arbitrary state to the maximally mixed one and that of the converse process happens to be their Kullback-Leibler relative entropy $D_{\tx{KL}}(\omega\|\mathbf q)$.

It is worth mentioning some directions of future work, which could complement the present results. First, since the characterization of the Finsler metric does not presuppose a particular norm, it would be interesting to extend these results to dynamics bounded on arbitrary norms on the \textit{individual} decay rates. Furthermore, it would be compelling to investigate situations in which the considered dynamics belong to a smaller subspace of generators, motivated by particular physical conditions or concrete experimental setups, possibly leading to results of more focused applicability.

Finally, we want to mention the quantum counterpart of the problem, namely, the description of the reachable set of quantum states under quantum Markovian dynamics. The quantum \textit{infinite-time} reachability problem has been studied in the context of quantum control theory, e.g. ~\cite{Cai2025,malvetti2024reachability,koch2022quantum} and,  notably, some of its features can be translated to the simplex, as noted in~\cite{malvetti2024reachability} and~\cite{schulte2019}. Similarly, some aspects of the quantum \textit{finite-time} reachability problem can be read from its classical analogue, see our successive publication~\cite{UsInPrep,PZ26}. Along these lines, it would be interesting to investigate whether the characterization given in this work implies any majorization conditions for quantum states, akin to those used in catalytic state transformations \cite{LipkaBartosik2024}.
\\

\normalfont
\ack{It is a pleasure to thank 
Dariusz Chru{\'s}ci{\'n}ski, 
Kamil Korzekwa,
Oskar Pro{\'s}niak and Fereshte Shahbeigi for numerous 
discussions which stimulated this work and fruitful interaction. 
We are grateful to Jan {\.Z}yczkowski for his help in preparation of Fig.~\ref{fig:riemann_finsler}.}

\funding{
~We acknowledge financial support by the European Union under ERC Advanced Grant TAtypic, Project No. 101142236. MH acknowledges that the study was funded by “Research support module” as part of the “Excellence Initiative - Research University” program at the Jagiellonian University in Kraków.}


\appendix

\section{Water redistribution model for bounded Kolmogorov dynamics}
\label{app: toy}

The discussed problem of optimal reachability with bounded generators can be intuitively rephrased as a task of water redistribution between deposits, using pumps that have limited draining power. Such hydraulic analogies are of frequent use within the analysis of Markov dynamics \cite{Harchol-Balter_2023}.

The analogy builds on thinking about states $p$ as specifying possible distributions of a fixed amount of water among $N$ tanks, with the entry $p_i$ indicating the filling fraction of the $i^\textrm{th}$ tank. The speed at which a deposit drains or fills can be read from the Kolmogorov equation,
\begin{equation}
    \frac{dp_i}{dt} = \underbrace{Q_{ii}p_i}_{\tx{Draining}} + \underbrace{\Large{\Sigma}_{j\neq i}Q_{ij}p_j}_{\tx{Filling}}
\end{equation}
We can think of each tank being drained by a corresponding pump, all of them converging into a common distributor from which return pipes controlled by valves share back the drained water into the tanks. Under isotropic norms only the diagonal entries of the generator are directly bounded and hence, from the point of view of resource theory, the draining cost of pump $i$ is proportional to $Q_{ii}$ while usage of return valves is free of cost. Note that the asymmetry of Kolmogorov dynamics is now apparent, since an empty tank can gain water from the return valves but, as the draining speed is proportional to the water level itself, a tank never empties in a finite time under fixed pump power $Q_{ii}$. In this setting, the finite-time reachability problem exactly translates to:

\

\noindent\textit{``Starting from an initial water configuration ${\bf{p}}$, what is the shortest time to achieve the target configuration ${\bf{q}}$ if the total pumping cost is limited, $\|Q\|_1^\tx{diag} = \sum_i|Q_{ii}|\leq 1$?''}

\

The following two illustrative examples make use of this hydraulic analogy to display the studied asymmetry between incoming and outgoing times in more familiar terms.

\setcounter{example}{0}
\begin{example}[Extremal $\ra$ flat water distribution]
\label{ex:water_1}
    Consider $N=3$ tanks. Initially, the first tank holds all the water $({\bf{p}}=[1,0,0]^T)$ and we wish to share it equally $({\bf{q}}=[1/3,1/3,1/3]^T)$ in the shortest time possible under the condition that the pumping cost is limited, $|Q_{11}| + |Q_{22}| + |Q_{33}|\leq 1$. Only the first tank is above its target, so just its draining pump needs to be activated. Setting $Q_{11}=1$ (the maximum possible power) and $Q_{22},Q_{33}=0$ and redistributing the collected water equally among tanks $2$ and $3$, the water distribution evolves as
    \begin{equation}
        {\bf{p}}(t) = [e^{-t},(1-e^{-t})/2,(1-e^{-t})/2]^T,
    \end{equation}
    achieving ${\bf{q}} = [1/3,1/3,1/3]^T$ at $t=\ln 3$, coinciding with the minimal time $T({\bf{p}},{\bf{q}})$ predicted by Theorem \ref{th: FT inner}. The resulting time is finite as expected, since no attempt to empty a tank has been made in the process.
    
\end{example}

\begin{example}[Flat $\ra$ extremal water distribution]
\label{ex:water_2}
    Consider now the converse situation, i.e. initially the water is evenly distributed $({\bf{q}}=[1/3,1/3,1/3]^T)$ and we wish to fill the first tank $({\bf{p}}_\varepsilon=[1-2\varepsilon,\varepsilon,\varepsilon]^T)$ in the shortest time possible under the condition that the total pumping cost is limited, $|Q_{11}| + |Q_{22}| + |Q_{33}|\leq 1$. In this case, the second and third tanks need to decrease their levels to $\varepsilon$, so putting $Q_{22}=Q_{33}=1/2$ and pouring the drained water into tank $1$, the water levels evolve as
    \begin{equation}
        {\bf{p}}(t) = [1-2e^{-t/2}/3,e^{-t/2}/3,e^{-t/2}/3]^T,
    \end{equation}
    achieving ${\bf{p}}_\varepsilon = [1-2\varepsilon,\varepsilon,\varepsilon]^T$ at $t=2\ln (1/3\varepsilon)$, which diverges as $\varepsilon\ra 0$: the emptier tanks 2 and 3 are required to become, the longer it takes. Note that this coincides with the minimal time $T({\bf{q}},{\bf{p}}_\varepsilon)$ predicted by Theorem \ref{th: FT inner}.
\end{example}
Fig.~\ref{fig:water_app} shows the evolutions corresponding to the previous examples in the probability simplex.

\begin{figure}[h!]
    \centering
\includegraphics[width=0.9\linewidth]{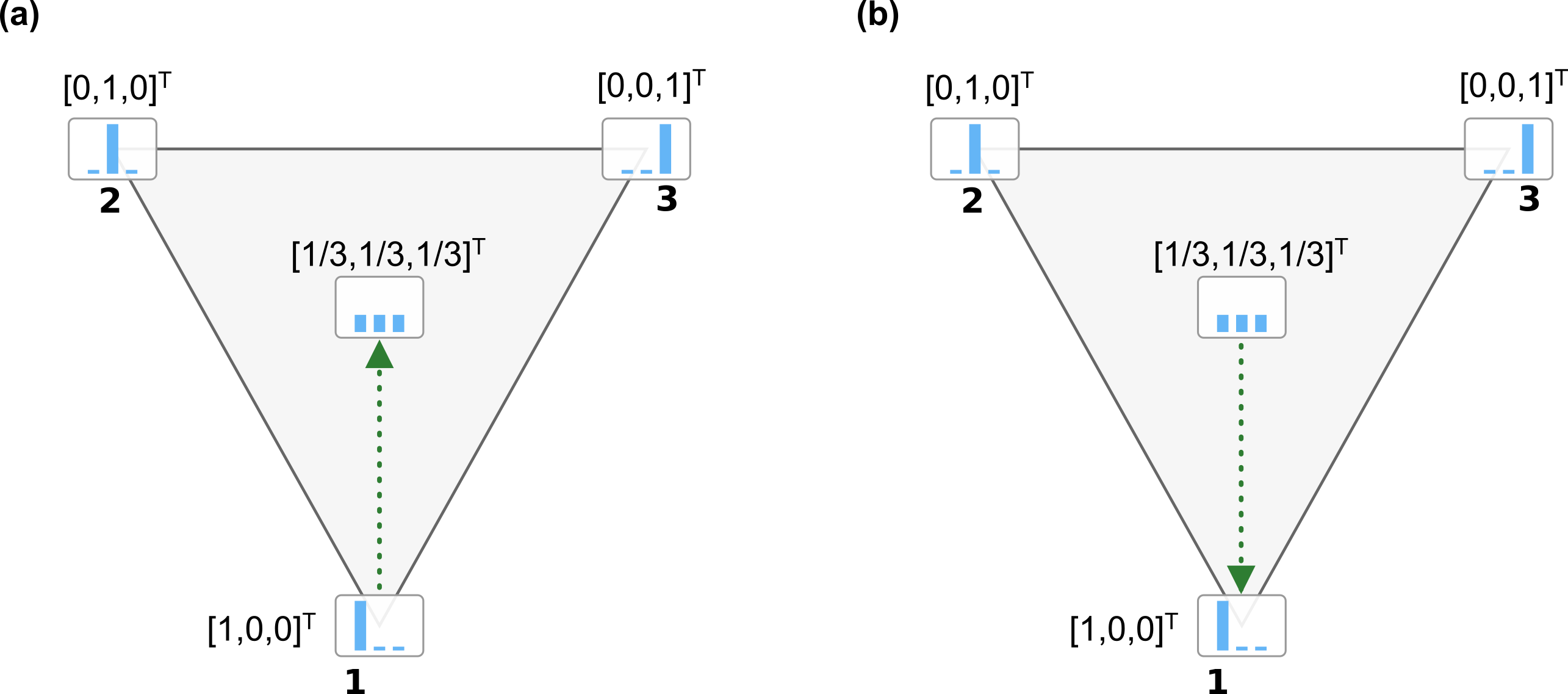}
    \caption{Green (dotted) arrow: evolution from the extremal state \([1,0,0]^T\) to the flat state \([1/3,1/3,1/3]^T\) in (a), and the reverse evolution in (b), corresponding to Examples~\ref{ex:water_1} and~\ref{ex:water_2}, respectively. The blue bars indicate the water level in each tank. }
    \label{fig:water_app}
\end{figure}


\section{Proofs of Section~\ref{sec:preliminaries}}
\label{app:prelim}

\noindent\bf{Proposition \ref{th: gen}}

\begin{proof}\label{pr: gen}
    An arbitrary Kolmogorov matrix $Q$ can be expanded as $Q = \sum_{ij}Q_{ij}E_{ij}$. Because of the constraint $Q_{ii} = -\sum_{j\neq i} Q_{ij}$, one has $Q = \sum_{i\neq j}Q_{ij} E_{ij}-\sum_{i\neq j}Q_{ij}E_{jj} = \sum_{i\neq j}Q_{ij}E_{(ij)}$. Since the $N(N-1)$ real numbers $\{Q_{ij}\}_{i\neq j}$ are non-negative, the previous equation defines $\mk{class}$ as a convex cone in $\mathbb{R}^{N(N-1)}$ spanned by the $\{E_{(ij)}\}_{i\neq j}$ and, since these are linearly independent (l.i.), the cone is simplicial.

    Observe that if, for $Q,Q'\in\mk{gen}$ we have $Q'-Q\in\mk{gen}$, then $Q'_{ij}-Q_{ij}\geq 0 \ \forall i\neq j$, so $Q'-Q\in\mk{gen} \then Q\leq_I Q'$. On the other hand, $Q\leq_I Q' \then Q'_{ij}- Q_{ij}\geq 0$ while $\sum_i(Q_{ij}'-Q_{ij})=0$ as $Q,Q'\in\mk{gen}$, so $Q\leq_IQ' \then Q'-Q\in\mk{gen}$. Hence, the relation $\leq_I$ is the partial order compatible with the cone $\mk{gen}$. However $\leq_C$, although reflexive and transitive, fails to be antisymmetric: for $N\geq 3$ there exist rate matrices with same collective decay rates but different individual ones, as
    \begin{equation}
        Q=
        \begin{bmatrix}
            -1 & . & .\\
            1 & -1 & .\\
            . & 1 & .
        \end{bmatrix} \tx{ and } Q'=
        \begin{bmatrix}
            -1 & . & .\\
            . & -1 & .\\
            1 & 1 & .
        \end{bmatrix}
    \end{equation}
    on $N=3$, so $\leq_C$ is just a pre-order on $\mk{gen}$. Still, it follows easily that $Q\leq_IQ' \then Q_{ij} \leq Q_{ij}' \forall i\neq j \then \sum_{i \neq j} Q_{ij} \leq \sum_{i \neq j}Q_{ij}' \forall i \then |Q_{ii}|\leq |Q_{ii}'| \then Q\leq_C Q'$. 
\end{proof}

\noindent\bf{Proposition \ref{th: gen(p,v)1}}

\begin{proof}\label{pr: gen(p,v)}
    Consider now two arbitrary interior states ${\bf{p}},{\bf{q}}$ and a given ${\bf{v}}$ \, \footnote{We really consider two vectors, ${\bf{v}}\in T_{\bf{p}}\Delta_{N-1}$ and its parallel transport to $T_{\bf{q}}\Delta_{N-1}$. However, the simplex is flat w.r.t. the standard connection in $\R^N$, parallel transport is path-independent and hence unique, and both ${\bf{v}}$ and its translation keep the same coordinates; so in an abuse of notation we refer to both vectors with the same symbol ${\bf{v}}$.}. The matrix $\mathsf{M}_{{\bf{p}}{\bf{q}}, \, ij} := \delta_{ij}p_i/q_i$ is invertible and we see that if $Q\in \mk{gen}({\bf{p}},{\bf{v}})$, then
    \begin{equation}
        [{Q\mathsf{M}_{\bf{pq}}}]_{ij} \geq 0 \ \forall i\neq j, \ \sum_{i=1}^N{[Q\mathsf{M}_{\bf{pq}}]}_{ij} = \sum_{i,k=0}^{N}Q_{ik}[{\mathsf{M}_{\bf{pq}}}]_{kj} = 1 \tx{ and } Q\mathsf{M}_{\bf{pq}}{\bf{q}} = Q{\bf{p}} = {\bf{v}}
    \end{equation}
    so $Q\mathsf{M}_{\bf{pq}} \in \mk{gen}({\bf{q}},{\bf{v}})$. That is, for invertible states and a given ${\bf{v}}$, the sets $\mk{gen}({\bf{p}},{\bf{v}})$ are all linearly isomorphic, and we may restrict our attention to $\mk{gen}(\omega,{\bf{v}})$ as some sort of ``canonical example''.
\end{proof}

\noindent\bf{Proposition \ref{th: gen(p,v)2}}

\begin{proof}
    
    When ${\bf{v}}={\bf{0}}$, $\mk{gen}(\mathbf{p},\mathbf{v})$ is just a cone ($Q\in \mk{gen}(\omega,{\bf{0}}) \then \lambda Q\in\mk{gen}(\omega,{\bf{0}})$ for $\lambda\geq 0$). 
    Being defined by the intersection of a linear subspace with $\mk{gen}$, the cone $\mk{gen}(\omega, {\bf{0}})$ is also pointed and inherits $\leq_I$ as its partial order, with $Q=0$ being its tip. Through the isomorphism $\ms{M}_{\omega{\bf{p}}}$ these conclusions hold for $\mk{gen}({\bf{p}},{\bf{0}})$ with arbitrary invertible ${\bf{p}}$.

    \
    
    For nonzero ${\bf{v}}$, the set $\mk{gen}(\omega,{\bf{v}})$ is the intersection of the $(N-1)^2$-dimensional affine subspace of matrices with column-sum 0 ($N$ constraints) and  $Q\omega={\bf{v}}$ ($N-1$ l.i. constraints), with the larger $\mk{gen}$, which ensures positivity. Their intersection is full-dimensional since, as can easily be seen, if $Q\in\mk{gen}(\omega,{\bf{v}})$, then the $(N-1)^2$-dimensional $Q+\mk{gen}(\omega,{\bf{0}})$ must also be contained. However, $\mk{gen}(\omega,{\bf{v}})$ is no longer a cone (not even affine-translated), as it has more than one extremal point: a counterexample in $N=4$ is given by the matrices $Q := 2E_{(24)} + E_{(31)} + 2E_{(34)} \tx{ and } Q' := E_{(24)} + E_{(21)} + 3E_{(34)}$, which are both extremal points of $\mk{gen}(\omega,{\bf{v}})$ with ${\bf{v}} = [-\frac{1}{4},\frac{1}{2},\frac{3}{4},-1]^T$. Thus, for nonzero ${\bf{v}}$, the partial order $\leq_I$ does not have a smallest element in $\mk{gen}(\omega,{\bf{v}})$. Still, we can show that the pre-order $\leq_C$ does have a set of smallest elements, i.e. matrices $Q$ such that $Q\leq_C Q'$ for any other $Q'\in\mk{gen}(\omega,{\bf{v}})$. Take, without loss of generality, ${\bf{v}}$ such that $v_i<0$ for $i\in\{1,2,...,n^-\}$ and the rest $v_i\geq 0$. An arbitrary $Q \in \mk{gen}(\omega,{\bf{v}})$ has the block decomposition
    \begin{equation}
        Q=
        \begin{bmatrix}
            -D & A\\
            A' & -D'
        \end{bmatrix}.
    \end{equation}
    with $D,D'$ having their diagonals non-negative, and all entries in $A,A'$ non-negative. The collective decay rates must be $Q_{ii} \leq Nv_i$ for $i\in I_- =\{1,...,n^-\}$ and $Q_{ii}\leq 0$ for $i\in I_+=\{n^-+1,...,N\}$. Then, the minimal choice w.r.t. $\leq_C$ is realized by generators of the form
    \begin{equation}
        Q =
        \begin{bmatrix}
            -D & 0\\
            WD & 0\\
        \end{bmatrix},
    \end{equation}
    where $D_{ij}=N\delta_{ij}|v_i|$ for $i,j\in I_-$ and $W_{ij}$ an entrywise nonnegative matrix satisfying $\sum_{j\in I_-}W_{ij}v_j = -v_i \ \forall i\in I_+$ and $\sum_{i\in I_+}W_{ij} = 1$. These conditions make the set of $W$ matrices a convex polytope of dimension $\dim(\R^{n^+n^-}) - (n^+ + n^- -1) = (n^+-1)(n^--1)$, coming from the fact that there is one linear dependency between the $n^++n^-$ constraints. Since any other generator $Q'$ in $\mk{gen}(\omega,{\bf{v}})$ will have diagonal entries satisfying $|Q_{ii}| \geq N|v_i|$ for $i \in I_-$, and $|Q_{ii}|\geq 0$ for $i \in I_+$, indeed we will have that $Q\leq_C Q'$ when $Q$ is of the above form. Finally, through the isomorphism $\ms{M}_{\omega{\bf{p}}}$ these conclusions hold for $\mk{gen}({\bf{p}},{\bf{v}})$ with arbitrary invertible ${\bf{p}}$, and we obtain the form of the minimal elements claimed on Proposition~\ref{th: gen(p,v)2}.
\end{proof}


\noindent\bf{Proposition \ref{th: T quasidistance}}

\begin{proof}\label{pr: T quasidistance}
    The first two conditions follow easily from the definition (\ref{eq: Tquasimetrics}). In addition, $T(\cdot,\cdot)$ has to obey the triangular inequality, as if not, the time-inhomogeneous families of rate matrices first leading from the initial ${\bf{p}}$ to the intermediate ${\bf{r}}$ and latter to ${\bf{q}}$, when concatenated, would constitute an evolution leading from ${\bf{p}}$ to ${\bf{q}}$ in a time smaller than $T({\bf{p}},{\bf{q}})$, contradicting the minimization in (\ref{eq: Tquasimetrics}). Note that the same proof does not apply if we had defined $T(\cdot,\cdot)$ via time-independent generators, as such a concatenation would not be a valid evolution.
\end{proof}

\noindent\bf{Theorem \ref{th: time norm}}

\begin{proof}\label{pr: time norm}
    We first show that $F({\bf{p}},{\bf{v}})$ is a $C^0$-Finsler metric, i.e. a (possibly asymmetric) norm on ${\bf{v}}$ for fixed ${\bf{p}}$, and continuous in ${\bf{p}}$. It is easy to compute the time quasi-distance between infinitesimally close states,
    \begin{equation}
        \lim_{h\ra 0^+}\frac{T({\bf{p}},{\bf{p}}+h{\bf{v}})}{h} = \min_Q \lambda~\tx{ s.t. } \ {\bf{v}}=Q{\bf{p}} \ \tx{ with } \ \|Q\|\leq \lambda \ \ = \!\!\!\! \min_{\mk{gen}({\bf{p}},{\bf{v}})} \!\! \|Q\| = F({\bf{p}},{\bf{v}}).
    \end{equation}
    Thus, from $T(\cdot,\cdot)$ being a quasi-distance (a standard distance except for the symmetry requirement), we infer that $F({\bf{p}},\cdot)$ obeys the axioms of an asymmetric norm on the tangent space of ${\bf{p}}$ (a standard norm except for the symmetry requirement):
    \begin{itemize}
        \item $F({\bf{p}},{\bf{v}})>0 \ \forall {\bf{v}}\neq {\bf{0}}$, following from the positivity of $T(\cdot,\cdot)$.

        \item $F({\bf{p}},\alpha {\bf{v}}) = \alpha F({\bf{p}},{\bf{v}})$ for positive $\alpha$, which follows from $h\mapsto h/\alpha$ in the above limit.

        \item $F({\bf{p}},{\bf{v}}+{\bf{w}}) \leq F({\bf{p}},{\bf{v}})+F({\bf{p}},{\bf{w}})$, due to the fact that $T(\cdot,\cdot)$ satisfies the triangular inequality.
    \end{itemize}
    We now need to show continuity in 
    in ${\bf{p}}$, which we do in two steps. Consider given ${\bf{p}},{\bf{v}}$. First, note that for any sequence of close states $\{{\bf{p}}_n\}\ra{\bf{p}}$ and some $Q\in\mk{gen}({\bf{p}},{\bf{v}})$, there exists\footnote{Consider the following construction. The valuation map $Q\mapsto\Phi_{\bf{p}}(Q):=Q{\bf{p}}$ is many-to-one, so take e.g. its Moore-Penrose pseudoinverse $\Phi_{\bf{p}}^+$. The matrices $Q_n:=Q-\Phi_{{\bf{p}}_n}^+(Q{\bf{p}}_n-{\bf{v}})$ all act as $Q_{n}{\bf{p}}_n={\bf{v}}$ and $\{Q_n\}\ra Q$.} a corresponding sequence of matrices $\{Q_n\}$ that converges to $Q$ and, thus, $\|Q_n\|\ra \|Q\|$. Taking $Q$ to be a minimal element in $\mk{gen}({\bf{p}},{\bf{v}})$, we conclude,
    \begin{equation}\label{eq: limsup}
        F({\bf{p}}_n,{\bf{v}})=\min_{\mk{gen}({\bf{p}}_n,{\bf{v}})}\|Q\|\leq \|Q_n\| \then \limsup_{n\ra\infty} F({\bf{p}}_n,{\bf{v}}) \leq \limsup_{n\ra\infty}\|Q_n\| = \|Q\| = F({\bf{p}},{\bf{v}}).
    \end{equation}
    On the other hand, consider again $\{{\bf{p}}_n\}\ra{\bf{p}}$ but let now $\{Q^*_n\}$ be a sequence of minimal elements in the corresponding sets $\mk{gen}({\bf{p}}_n,{\bf{v}})$. Let $L=\liminf_{n\ra\infty}\|Q^*_n\|$, and let $\{Q_{n_k}^*\}$ be the subsequence that realizes it. It converges in norm to $L$, so it is bounded and the Heine-Borel theorem assures the existence of a further subsequence $\{Q_{n_{k'}}^*\}$ which converges to some $Q^*$ with norm $\|Q^*\|=L$ (which lies in $\mk{gen}({\bf{p}},{\bf{v}})$, as $\lim_{k'\ra\infty}Q^*_{n_{k'}}=Q^*$ implies $\lim_{k'\ra\infty}Q^*_{n_{k'}}{\bf{p}}=Q^*{\bf{p}}=\bf{v}$). Hence,
    \begin{equation}\label{eq: liminf}
        \liminf_{n\ra\infty} F({\bf{p}}_n,{\bf{v}}) = \liminf_{n\ra\infty} \|Q^*_n\| =L = \lim_{k'\ra\infty}\|Q^*_{n_k'}\|=\|Q^*\|\geq \min_{\mk{gen}({\bf{p}},{\bf{v}})}\|Q\| = F({\bf{p}},{\bf{v}}).
    \end{equation}
    From \eqref{eq: limsup} and \eqref{eq: liminf} we conclude $\lim_{n\ra\infty}F({\bf{p}}_n,{\bf{v}}) = F({\bf{p}},{\bf{v}})$ and, hence, continuity on ${\bf{p}}$.

    \

    Having proven that $F$ is indeed a Finsler metric, let us address its compatibility with $T$. Consider the geodesic minimization over paths $\gamma:[0,1]\ra \Delta_{N-1}$ with ${\bf{p}}$ and ${\bf{q}}$ as endpoints.
    \begin{align}
        \min_{\substack{\gamma \ \tx{s.t.}\\ \gamma(0)={\bf{p}} \\ \gamma(1) = {\bf{q}}}} L[\gamma] = \min_{\substack{\gamma \ \tx{s.t.}\\ \gamma(0)={\bf{p}} \\ \gamma(1) = {\bf{q}}}} \int_0^1 F(\gamma(s),\dot\gamma(s))
    \end{align}
    The length functional is reparameterization-invariant so, in principle, we may parameterize each curve, $\bar \gamma(t) = \gamma(s(t))$, in such a way that $F(\bar\gamma(t),d\bar\gamma/dt)=1$, i.e. all tangent vectors have Finsler norm 1 in the new parameterization. Thus, each curve now reaches the target state ${\bf{q}}$ at a particular parameter unit $t$, and the geodesic minimization becomes
    \begin{align}
        \min_{\substack{\bar\gamma \ \tx{s.t.}\\\bar\gamma(0)={\bf{p}} \\ \bar\gamma(t)={\bf{q}}}} \! t \ \tx{ subject \ to } \!\!\!\!\!\!\!\!\!\ \min_{\mk{gen}(\bar\gamma(s),\dot{\bar\gamma}(s))} \!\!\!\!\!\!\!\! \|Q\| = 1
    \end{align}
    Since any curve can be reproduced by a suitable time-dependent Kolmogorov matrix, let us reformulate the problem as
    \begin{align}
        \min_{\substack{Q_s \ \tx{s.t.}\\{\bf{q}} = \mc{T}e^{\int_0^tQ_sds}{\bf{p}}}} \!\!\!\!\!\! t \ \tx{ subject \ to } \!\!\!\!\!\!\!\!\!\ \min_{\mk{gen}({\bf{p}}_s,Q_s({\bf{p}}_s)} \!\!\!\!\!\!\!\! \|Q\| = 1
    \end{align}
    Now, the condition $\min_{\mk{gen}({\bf{p}}_s,Q_s({\bf{p}}_s)} \|Q\| = 1$ imposes that necessarily $\|Q_s\|\geq 1 \forall s$. Furthermore, if we assume that the time-dependent generator that solves the problem has $\|Q_s\|>1$ during some interval, we arrive at a contradiction (define $\bar{Q}_s:=Q_s/\|Q_s\|$, joining ${\bf{p}}$ and ${\bf{q}}$ in a smaller time), so minimizers have $\|Q_s\|=1$, and this last condition may be imposed to the problem without modifying its solutions,
    \begin{align}\label{eq: geodesic_min_time}
        \min_{\substack{\gamma \ \tx{s.t.}\\ \gamma(0)={\bf{p}} \\ \gamma(1) = {\bf{q}}}} L[\gamma] =  \!\!\!\!\!\!\!\!\! 
        \min_{\substack{Q_s \ \tx{s.t.}\\{\bf{q}} = \mc{T}e^{\int_0^tQ_sds}{\bf{p}}}} \!\!\!\!\!\! t \ \tx{ subject \ to } \!\!\!\!\!\!\!\!\!\ \min_{\mk{gen}({\bf{p}}_s,Q_s({\bf{p}}_s)} \!\!\!\!\!\!\!\! \|Q\| = 1 \ \tx{ and } \ \|Q_s\|=1.
    \end{align}
    On the other hand, if we consider the time minimization problem that defines the time quasi-distance $T({\bf{p}},{\bf{q}})$,
    \begin{equation}
        T({\bf{p}}, {\bf{q}}):= \!\!\!\!\!\!\!\!\! \min_{\substack{Q_s \ \tx{s.t.}\\{\bf{q}} = \mc{T}e^{\int_0^tQ_sds}{\bf{p}}}} \!\!\!\!\!\! t \ \tx{ subject \ to } \ \|Q_s\| = 1,
    \end{equation}
    it is easy to see that the normalization condition implies $\min_{\mk{gen}({{\bf{p}}_s,Q({\bf{p}}_s)})}\|Q\|\leq 1$. If we assume the solution $Q_s$ to have $\min_{\mk{gen}({{\bf{p}}_s,Q({\bf{p}}_s)})}\|Q\|< 1$ we arrive at a contradiction (there would exist a $Q_s'\in\mk{gen}({\bf{p}}_s,Q_s({\bf{p}}_s))$ giving exactly the same trajectory, but with a norm $\|Q_s'\|<1$. Hence, $Q_s'':=Q_s'/\|Q_s'\|$ would be a Kolmogorov generator satisfying all constrains and joining ${\bf{p}},{\bf{q}}$ in a time no larger than the putative minimizer $Q_s$). Thus we conclude that the solutions to this minimization obey $\min_{\mk{gen}({{\bf{p}}_s,Q({\bf{p}}_s)})}\|Q\|= 1$, and this additional condition can be imposed without harm. But then the time minimization problem takes precisely the form \eqref{eq: geodesic_min_time} of the geodesic problem. Hence, both the geodesic and the time-minimization problems are equivalent, and
    \begin{equation}
        T({\bf{p}},{\bf{q}}) = \min_{\substack{\bar\gamma \ \tx{s.t.}\\\bar\gamma(0)={\bf{p}} \\ \bar\gamma(1)={\bf{q}}}} L[\bar\gamma],
    \end{equation}
    that is, the geodesic quasi-distance defined by the Finsler metric $F({\bf{p}},{\bf{v}})$ is $T({\bf{p}},{\bf{q}})$.
\end{proof}

\section{Proofs of Section~\ref{sec:results}}
\label{app:results}



\noindent\bf{Proposition \ref{th: TS outer}}

\begin{proof}\label{pr: TS outer}
    As seen, all Finsler metrics induced by isotropic norms $\|\cdot\|$ of same time-scale $\tau$ behave equally on the single-negative-entry (say, $v_i$) chambers of $T_\textbf{p}\Delta_{N-1}$, reducing to $F({\bf{p}},{\bf{v}}) = \tau |v_i|/p_i$, and so, all leading to the same largest velocities on this chambers: $v_i = -p_i/\tau$ and other entries arbitrary as long as non-negative and adding up to $|v_i|$. The inner and outer bounds will be derived as the ``tightest'' and ``largest'' convex sets consistent with these constraints. First, note that the extremal velocities in the negative-$i$ chamber will have only a single positive entry (say $v_j$) and the $N-2$ remaining being 0. Denoting it as ${\bf{v}}_{(ij)}$, it has components
    \begin{equation}
        {v_{(ij)}}_k :=
        \begin{cases}
            -p_i/\tau & \tx{ if } i=k \\
            p_i/\tau & \tx{ if } j=k\\
            0 & \tx{ otherwise}
        \end{cases}
    \end{equation}
    There are $N-1$ such extremals for each of the $N$ single-negative-entry chambers. That the convex hull of these $N(N-1)$ extremal velocities forms an inner bound to $\mc{J}_\textbf{p}$ follows by convexity of the latter. Now, the Finsler metric associated to $\|Q\|_1^\tx{diag}$ is $F_1({\bf{p}},{\bf{v}})=\tau\sum_{v_i\leq 0}|v_i|/p_i$, and it is easy to see that all the extremal vectors of $\mc{J}^\tx{inner}_{\bf{p}}$ have $F({\bf{p}},{\bf{v}}) = 1$ and, due to the linearity of $F$, we conclude that their convex hull coincides with its unit ball -- the set of achievable velocities under $\|Q\|^\tx{diag}_1\leq 1$. On the other hand, we may consider the $N$ half spaces $v_i \geq -p_i/\tau$ that support the largest velocities at each of the $N$ single-negative-entry chamber. Their intersection,
    \begin{equation}
        \mc{J}^\tx{outer}_{\bf{p}} := \{{\bf{v}}\in T_{\bf{p}}\Delta_{N-1} \tx{ s.t. } v_i\geq -p_i/\tau\} = \frac{1}{\tau}(\Delta_{N-1}-{\bf{p}}),
    \end{equation}
    which supports $\mc{J}_{\bf{p}}$ is, hence, the largest convex set enclosing it. The Finsler metric associated to $\|Q\|_\infty^\tx{diag}$ is $F_\infty({\bf{p}},{\bf{v}})=\tau\max_{v_i\leq 0}|v_i|/p_i$, and it is easy to see that $F({\bf{p}},{\bf{v}})\leq 1$ for vectors ${\bf{v}}\in\mc{J}^\tx{outer}_{\bf{p}}$ and, thus, its unit ball (i.e. the achievable velocities under $\|Q\|^\tx{diag}_\infty\leq 1$) encloses or equals the latter set. Since $\mc{J}^\tx{outer}_{\bf{p}}$ is shown to be outer bound for any isotropic norm of time-scale $\tau$, we conclude they must be equal.
    
\end{proof}

\noindent\bf{Theorem \ref{th: FT inner}}

\begin{proof}\label{pr: FT inner}
    First, consider the function $T_\alpha({\bf{p}},{\bf{q}}):=\tau\left(\sum_{i=1}^N\max\left(0,\ln\frac{p_i}{q_i}\right)^\alpha\right)^{1/\alpha}$. We will prove it is indeed a quasi-distance on the simplex $\Delta_{N-1}$:
    \begin{itemize}
        \item \textit{Non‑negativity} is ensured by construction.

        \item \textit{Point-separability:} If $\textbf{q} = \textbf{p}$, then $\ln\frac{p_i}{q_i} = 0 \ \forall i$; so $T_\alpha(\textbf{q},\textbf{p})=0$. And if $T(\textbf{q},\textbf{p})=0$, since each term in the sum (\ref{eq: InnerFinslergen}) is non‑negative, they must have $\ln\frac{p_i}{q_i}=0$, i.e.\ $p_i = q_i$, for every $i$ with $q_i \le p_i$. For indices where $q_i > p_i$, we would have $p_i - q_i < 0$ but, since $\sum_i p_i = \sum_i q_i = 1$, it follows that $\sum_{i: q_i > p_i} (p_i - q_i) = 0$, which forces the set $\{i : q_i > p_i\}$ to be empty. Hence $q_i = p_i$ for all $i$, and $\textbf{q} = \textbf{p}$.

        \item \textit{Triangle inequality:} For any $\textbf{p},\textbf{r},\textbf{q} \in \Delta_{N-1}$,
        \begin{align}\nonumber
            &T_\alpha({\bf{p}},{\bf{q}})=\tau\left(\sum_{i=1}^N\max\left(0,\ln\frac{p_i}{q_i}\right)^\alpha\right)^{1/\alpha} = \tau\left(\sum_{i=1}^N\max\left(0,\ln\frac{p_i}{r_i} + \ln\frac{r_i}{q_i}\right)^\alpha\right)^{1/\alpha}\\
            &\leq \tau\left(\sum_{i=1}^N\max\left(0,\ln\frac{p_i}{r_i}\right)^\alpha\right)^{1/\alpha} + \tau\left(\sum_{i=1}^N\max\left(0,\ln\frac{r_i}{q_i}\right)^\alpha\right)^{1/\alpha} = T_\alpha(\textbf{p},\textbf{r}) + T_\alpha(\textbf{r},\textbf{q}).
        \end{align}
    \end{itemize}
    Furthermore, for infinitesimally close states
    \begin{equation}
        \lim_{h\ra 0^+} \frac{T_\alpha(\textbf{p},\textbf{p} + h\textbf{v})}{h} = \tau \lim_{h\ra 0^+} \frac{1}{h}\left(\sum_{i \tx{ s.t. } v_i\leq 0} \!\!\!\!\! \ln(1+hv_i/p_i)^p \right)^{1/\alpha}= \tau\left(\sum_{v_i\leq 0}\left|\frac{v_i}{p_i}\right|^\alpha\right)^{1/\alpha},
    \end{equation}
    One can see that, according to \eqref{eq: TFinslernorm}, this Finsler metric is precisely that due to dynamics bounded by $\|Q\|_\alpha^\tx{ diag}:=\tau\left(\sum_i |Q_{ii}|^\alpha\right)^{1/\alpha}\leq1$. 
    But from this it does not directly follow that the geodesic quasi-distance of this $F$ coincides with the proposed $T_\alpha(\cdot,\cdot)$, as two different quasi-distances may share the same infinitesimal behavior. To show their equality, we have to prove that $T(\bf{p},\bf{q})$ is indeed the length of the shortest path joining $\bf{p}$ to $\bf{q}$ according to $F$. For arbitrary initial and target interior states, we can see that any curve $\gamma$ (parameterized so that it departs ${\bf{p}}$ at $s=0$ and reaches ${\bf{q}}$ at $s=1$) has a length
    \begin{align}
        L[\gamma] &= \int_0^1 F(\bar\gamma(s),\dot{\bar\gamma}(s)) = \tau \int_0^1 \left(\sum^N_{i=1}\max\left(0,-\frac{\dot p_i(s)}{p_i(s)}\right)^\alpha\right)^{1/\alpha}\\ &\geq \tau\left(\sum_{i=1}^N\max\left(0,-\int_0^1 \frac{\dot p_i(s)}{p_i(s)}\right)^\alpha\right)^{1/\alpha} = \tau\left(\sum_{i=1}^N\max\left(0,\ln\frac{p_i}{q_i}\right)^\alpha\right)^{1/\alpha} = T_\alpha({\bf{p}},{\bf{q}}).
    \end{align}
    where Minkowski's integral inequality was used. Hence, the $T_\alpha$ quasi-distance establishes a lower bound to the length of directed curves joining pairs of states. Now, we note that there exist semigroup evolutions that saturate this lower bound. Indeed, choose any initial and target interior states ${\bf{p}},{\bf{q}}$, let $I_-=\{i \tx{s.t.} q_i<p_i\}$ and $I_+:=\{i \tx{s.t.} q_i\geq p_i\}$ and, for simplicity, order the basis so that $I_-$-components appear first and $I_+$-components later. Construct the generator
    \begin{equation}
        \bar Q:=
        \begin{bmatrix}
            -D & 0\\
            WD & 0
        \end{bmatrix}
    \end{equation}
    where the $I_-\times I_-$ block $D_{ij}=\delta_{ij}\ln(p_i/q_i)$ and the $I_+\times I_-$ block $W_{ij}$ is an entrywise non-negative matrix such that $\sum_{j_\in I_-}W_{ij}(p_j-q_j) = q_i-p_i$ and $\sum_{i\in I_+}W_{ij}=1$. Its semigroup acts as
    \begin{equation}
        e^{\bar Q} = 
        \begin{bmatrix}
            e^{-D} & 0\\
            W(\mathbb{I}-e^{-D}) & \mathbb{I}
        \end{bmatrix} \then e^{\bar Q}{\bf{p}} =
        \begin{bmatrix}
            e^{-D} & 0\\
            W(\mathbb{I}-e^{-D}) & \mathbb{I}
        \end{bmatrix}
        \begin{bmatrix}
            p_-\\
            p_+
        \end{bmatrix} = 
        \begin{bmatrix}
            q_-\\
            W(q_--p_-) + p_+
        \end{bmatrix} = 
        \begin{bmatrix}
            q_-\\
            q_+
        \end{bmatrix} = {\bf{q}}.
    \end{equation}
    Since $\bar Q$ joins them in unit time, the normalized generator $Q:=\bar Q/\|\bar Q\|^\tx{diag}_\alpha$ -- giving the fastest parameterization of the semigroup path $\gamma$ allowed by the constraint $\|\cdot\|^\tx{diag}_\alpha\leq 1$ -- takes a time $L[\bar \gamma] = \|\bar Q\|^\tx{diag}_\alpha = \tau(\sum_{i\in I_-}\ln(q_i/p_i)^\alpha)^{1/\alpha}$. But this is exactly the quasi-distance $T_\alpha({\bf{p}},{\bf{q}})$ between them. Therefore, we conclude that $T_\alpha({\bf{p}},{\bf{q}}) = \inf_\gamma L[\gamma]$ over all directed curves joining ${\bf{p}}$ to ${\bf{q}}$ -- that the proposed $T(\cdot,\cdot)$ is the time quasi-distance arising from dynamics bounded as $\|Q\|_\alpha^\tx{diag}\leq 1$ -- and that among the fastest allowed evolutions are the above described semigroups.

\end{proof}

\noindent\textbf{Corollary \ref{th: FT outer}}

\begin{proof}\label{pr: FT outer}
    Consider an arbitrary isotropic norm, inducing the Finsler metric $F$ and corresponding time quasi-distance $T$. An equivalent statement of Proposition \ref{th: TS outer} is that $F_\infty({\bf{p}},{\bf{v}}) \leq F({\bf{p}},{\bf{v}}) \leq F_1({\bf{p}},{\bf{v}})$, where $F_1$ and $F_\infty$ denote, respectively, the Finsler metrics induced by dynamics constrained in the norms $\|\cdot\|_1^\tx{diag}$ and $\|\cdot\|^\tx{diag}_\infty$. It directly follows that    
    \begin{equation}
         T_\infty({\bf{p}},{\bf{q}}) \leq T({\bf{p}},{\bf{q}}) \leq T_1({\bf{p}},{\bf{q}}) \tx{ and }  \mc{S}_t^{1\pm}({\bf{p}}) \subset \mc{S}_t^\pm({\bf{p}}) \subseteq\mc{S}_t^{\infty\pm}({\bf{p}})
    \end{equation}
    where $T_1$ and $T_\infty$ denote the appropriate quasi-distances, with forms as given in Theorem \ref{th: FT inner}. The quasi-distance $T_\infty$ happens to be the well established Funk quasi-distance \cite{Funk1929,Papadopoulos2014}, which on the simplex is equivalent to the $\infty$-Rényi divergence, and its reachable sets take the form presented in \eqref{eq:sets_outer}.
\end{proof}

\noindent\textbf{Theorem~\ref{prop: threshold}}
\begin{proof}\label{app:threshold-proof}
Let $X_1,X_2,\ldots$ be independent standard exponential random
variables and, for every $N\geq2$, define
\begin{equation}
S_N:=\sum_{j=1}^N X_j,
\qquad
\mathbf q^{(N)}:=\frac{1}{S_N}(X_1,\ldots,X_N).
\label{eq: exponential-simplex-representation}
\end{equation}
Then $\mathbf q^{(N)}\sim\tx{Dir}(1,\ldots,1)$ and is uniformly
distributed on $\Delta_{N-1}$. Put $a_N:=S_N/N$. By the strong law of
large numbers, $a_N\to1$ almost surely. Since
$Nq_i^{(N)}=X_i/a_N$ and $\tau=1/N$,
\begin{equation*}
\begin{aligned}
T_N
:=T_1\bigl(\omega_N,\mathbf q^{(N)}\bigr)
=\frac{1}{N}\sum_{i=1}^N
\max\left\{0,-\ln\frac{X_i}{a_N}\right\}.
\end{aligned}
\end{equation*}
The map $x\mapsto\max\{0,x\}$ is $1$-Lipschitz, so
\begin{equation*}
\left|
T_N-\frac{1}{N}\sum_{i=1}^N\max\{0,-\ln X_i\}
\right|
\leq|\ln a_N|
\xrightarrow[N\to\infty]{\mathrm{a.s.}}0.
\end{equation*}
The variables $Y_i:=\max\{0,-\ln X_i\}$ are independent, identically
distributed, and integrable, with
\begin{equation}
\mathbb E[Y_1]
=\int_0^1(-\ln x)e^{-x}\,dx
=\gamma+\mathrm E_1(1)
=t_c.
\label{eq: threshold-expectation}
\end{equation}
A second application of the strong law gives
$N^{-1}\sum_{i=1}^N Y_i\to t_c$ almost surely. Consequently,
$T_N\to t_c$ almost surely under this coupling and hence in probability
with respect to the uniform target distribution.
It remains to verify the convergence of the expectation appearing in
Eq.~\eqref{eq: threshold}. Define $Z_N:=Nq_1^{(N)}$. The first marginal
of the uniform Dirichlet distribution $\tx{Dir}(1,\ldots,1)$ satisfies
~$q_1^{(N)}\sim\operatorname{Beta}(1,N-1)$, with density
$(N-1)(1-u)^{N-2}$ for $0<u<1$. Therefore, by the change of variables
$x=Nu$, the density of $Z_N$ is
\begin{equation}
\rho_N(x)
=\frac{N-1}{N}
\left(1-\frac{x}{N}\right)^{N-2},
\qquad 0<x<N.
\end{equation}
By exchangeability of the components of $\mathbf q^{(N)}$, we establish the average
\begin{align}
\mathbb E[T_N]
=\frac{1}{N}\sum_{i=1}^N
\mathbb E\!\left[\max\{0,-\ln(Nq_i^{(N)})\}\right]
=\mathbb E\!\left[\max\{0,-\ln Z_N\}\right]
=\int_0^1(-\ln x)\rho_N(x)\,dx.
\end{align}
For every $x\in(0,1)$,
$\rho_N(x)\xrightarrow[N\to\infty]{}e^{-x}$, while
$0\leq\rho_N(x)\leq1$. Since $-\ln x$ is integrable on $(0,1)$,
the dominated convergence theorem shows that $\mathbb E[T_N]$
converges to the integral evaluated in
Eq.~\eqref{eq: threshold-expectation}. Therefore,
\begin{equation}
\lim_{N\to\infty}\mathbb E[T_N]
=\mathbb E[Y_1]
=t_c.
\end{equation}
This proves the mean convergence stated in Eq.~\eqref{eq: threshold}. Moreover,
$\frac{V^+_{\tx{inner}}(t)}{V_{\tx{total}}}
=\Pr\!\left[T_N\leq t\right].$
Thus the inner-volume fraction tends to $0$ for fixed $t<t_c$ and to
$1$ for fixed $t>t_c$. Since
\begin{equation}
\frac{V^+_{\tx{outer}}(t)}{V_{\tx{total}}}
=(1-e^{-Nt})^{N-1}
\xrightarrow[N\to\infty]{}1
\label{eq: outer-volume-large-N}
\end{equation}
for every fixed $t>0$, division by the outer-volume fraction proves
Eq.~\eqref{eq: tightness-threshold}.
\end{proof}

\begin{proposition}[Incoming covering time]
\label{prop: saturation}
Let $N\geq2$, let $\omega_N=(1/N,\ldots,1/N)^T$, and let
$T:=T_1$ be the exact minimal-time function for the
diagonal-$\ell_1$ norm with $\tau=1/N$. Define
\begin{equation}
K_N^*:=\operatorname*{arg\,max}_{k\in\{1,\ldots,N-1\}}
\frac{k}{N}\ln\frac Nk,
\qquad
k_N^*\in K_N^*,
\label{eq: kNstar-definition}
\end{equation}
where $k_N^*$ is any chosen maximizer, and set
\begin{equation}
t_N^*:=\frac{k_N^*}{N}\ln\frac{N}{k_N^*}
=\max_{1\leq k<N}\frac{k}{N}\ln\frac Nk.
\label{eq: tNstar-definition}
\end{equation}
The maximizing integer need not be unique, but the value $t_N^*$ is
unique. Then
\begin{equation}
\max_{\mathbf q\in\Delta_{N-1}}T(\mathbf q,\omega_N)=t_N^*,
\label{eq: Dinmax}
\end{equation}
and the complete set of maximizers is
\begin{equation}
\begin{aligned}
\operatorname*{arg\,max}_{\mathbf q\in\Delta_{N-1}}
T(\mathbf q,\omega_N)
=\bigcup_{k\in K_N^*}\Bigl\{\frac1k\mathbf 1_I:
 I\subseteq\{1,\ldots,N\},\ |I|=k\Bigr\}.
\end{aligned}
\label{eq:incoming-maximizers}
\end{equation}
Here $\mathbf 1_I$ denotes the indicator vector of $I$. Consequently,
\begin{equation}
\begin{aligned}
\mathcal S_t^-(\omega_N)
&:=\left\{\mathbf q\in\Delta_{N-1}:
T(\mathbf q,\omega_N)\leq t\right\},\\
\mathcal S_t^-(\omega_N)=\Delta_{N-1}
&\quad\Longleftrightarrow\quad t\geq t_N^*.
\end{aligned}
\label{eq:diagonal-l1-saturation}
\end{equation}
Furthermore, for every sequence of choices $k_N^*\in K_N^*$,
\begin{equation}
t_N^*\leq\frac1e,
\qquad
\frac{k_N^*}{N}\xrightarrow[N\to\infty]{}\frac1e,
\qquad
t_N^*\xrightarrow[N\to\infty]{}\frac1e.
\label{eq: tstar}
\end{equation}
For any other isotropic norm $\|\cdot\|$ with the same time scale,
$t_N^*$ is an upper bound on the incoming covering time, but need not
be its exact value.
\end{proposition}

\begin{proof}
By~\eqref{eq: outgoing-incoming-costs},
\begin{equation*}
T(\mathbf q,\omega_N)
=\frac1N\sum_{i\in I_+}\ln(Nq_i),
\qquad
I_+:=\{i:q_i>1/N\}.
\end{equation*}
Write $k:=|I_+|$. If $k=0$, normalization gives
$\mathbf q=\omega_N$ and $T(\mathbf q,\omega_N)=0$; $k=N$ is
impossible because it would imply $\sum_iq_i>1$. For $1\leq k<N$, put
$m:=\sum_{i\in I_+}q_i\leq1$. Jensen's inequality gives
\begin{equation}
N T(\mathbf q,\omega_N)
=\sum_{i\in I_+}\ln(Nq_i)
\leq k\ln\frac{Nm}{k}
\leq k\ln\frac Nk.
\label{eq:incoming-jensen-chain}
\end{equation}
For a fixed $k$, equality throughout holds if and only if $m=1$ and
$q_i=1/k$ on $I_+$, equivalently
$\mathbf q=k^{-1}\mathbf 1_I$ for some $I$ with $|I|=k$. Maximizing
over $k\in\{1,\ldots,N-1\}$ proves~\eqref{eq: Dinmax}; the same
equality conditions give~\eqref{eq:incoming-maximizers}. The
sublevel-set definition then gives~\eqref{eq:diagonal-l1-saturation}.

For the asymptotic statement, extend $f(x):=-x\ln x$ continuously to
$[0,1]$ by setting $f(0)=f(1)=0$. Then
\begin{equation}
t_N^*=\max_{1\leq k<N}f\!\left(\frac{k}{N}\right).
\label{eq: discrete-envelope}
\end{equation}
The function $f$ has a unique maximum at $x=1/e$, with
$f(1/e)=1/e$, so $t_N^*\leq1/e$. Choose integers
$j_N\in\{1,\ldots,N-1\}$ with $j_N/N\to1/e$. Since
\begin{equation*}
f(j_N/N)\leq t_N^*\leq1/e,
\end{equation*}
we have $t_N^*\to1/e$. If a sequence $k_N^*\in K_N^*$ did not satisfy
$k_N^*/N\to1/e$, compactness would give a subsequence converging to
some $x\neq1/e$. Continuity would then imply
$f(x)=\lim_Nt_N^*=1/e$, contradicting uniqueness of the maximizer.
Therefore $k_N^*/N\to1/e$.

Finally, for every isotropic norm $\|\cdot\|$ with the same time scale,
$\|Q\|\leq\|Q\|_1^{\mathrm{diag}}$. Thus every control admissible for
the diagonal-$\ell_1$ norm is also admissible for $\|\cdot\|$, proving
the stated upper bound on the incoming covering time.
\end{proof}

%
%




\normalfont

\bibliographystyle{iopart-num-long}

\bibliography{Biblio}
\end{document}